\documentclass[aps,prb,floatfix,longbibliography,twocolumn,footinbib,superscriptaddress]{revtex4-2}  
\expandafter\let\csname equation*\endcsname\relax
\expandafter\let\csname endequation*\endcsname\relax

\usepackage{graphicx}  
\usepackage{mathtools}
\usepackage{amsfonts, amssymb}   
\usepackage{bm}        
\usepackage{amsmath}    
\usepackage{amsthm}    
\usepackage{verbatim}   
\usepackage{xcolor}      
\usepackage[ruled,vlined]{algorithm2e}
\usepackage[caption=false]{subfig}  
\usepackage{natbib}
    \allowdisplaybreaks
\usepackage[colorlinks=true, citecolor=magenta, linkcolor=blue, filecolor=magenta, urlcolor=blue]{hyperref}   
\usepackage{float}
\usepackage{physics}
\usepackage{enumerate}
\usepackage{mathrsfs}
\usepackage{relsize}
\usepackage{wasysym}
\usepackage{tikz}
\usetikzlibrary{graphs, arrows,chains,matrix,positioning,scopes,shadows, calc}
\usepackage{scalerel,stackengine}
\usepackage{accents}
\usepackage{orcidlink}
\usepackage{soul}
\usepackage[normalem]{ulem}

\stackMath
\newcommand{\stretchedhat}[1]{%
\savestack{\tmpbox}{\stretchto{%
  \scaleto{%
    \scalerel*[\widthof{\ensuremath{#1}}]{\kern.1pt\mathchar"0362\kern.1pt}%
    {\rule{0ex}{\textheight}}
  }{\textheight}%
}{2.4ex}}%
\stackon[-6.9pt]{#1}{\tmpbox}%
}
\newcommand{\stretchedtilde}[1]{%
\savestack{\tmpbox}{\stretchto{%
  \scaleto{%
    \scalerel*[\widthof{\ensuremath{#1}}]{\kern.1pt\mathchar"307E\kern.1pt}%
    {\rule{0ex}{\textheight}}
  }{\textheight}%
}{2.4ex}}%
\stackon[-6.9pt]{#1}{\tmpbox}%
}
\newtheorem{theorem}{Theorem}

\newtheorem{corollary}{Corollary}
\newtheorem{lemma}{Lemma}

\theoremstyle{definition}

\newcommand{\half}{\frac{1}{2}}

\newcommand{\imi}{\mathrm{i}}
\newcommand{\qexp}[3]{\vphantom{#2}\left\langle{#1}\middle\vert\smash{#2}\middle\vert{#3}\right\rangle}

\newcommand{\Hil}{{\mathcal H}}

\newcommand{\Th}{\text{th}}

\newcommand{\dprime}{\prime\prime}

\newcommand{\tikztriangle}{\raisebox{-0.3ex}{\begin{tikzpicture}
    [shift={(0,0)}]
    \coordinate (A) at (0,0);
    \coordinate (B) at (.3,0);
    \coordinate (C) at ($(A)!0.5!(B) + (0, {0.15*sqrt(3)})$);
    \fill[blue] (A) -- (B) -- (C) -- cycle;
    \draw[-] (A) -- (B) -- (C) -- cycle;
\end{tikzpicture}}}
\newcommand{\tikzsquare}{\raisebox{-0.3ex}{\begin{tikzpicture}
    [shift={(0,0)}]
    \coordinate (A) at (0,-0.15);
    \coordinate (B) at (0.15,0);
    \coordinate (C) at (0,.15);
    \coordinate (D) at (-.15,0);
    \fill[red] (A) -- (B) -- (C) -- (D) -- cycle;
    \draw[-] (A) -- (B) -- (C) -- (D) -- cycle;
\end{tikzpicture}}}

\newcommand{\bmu}{\bm{U}}
\newcommand{\gsun}[1]{\mathrm{SU}(#1)}

\newcommand{\pcsadd}{Center for Theoretical Physics of Complex Systems, Institute for Basic Science, Daejeon 34126, Republic of Korea}
\newcommand{\ustadd}{Basic Science Program, Korea University of Science and Technology, Daejeon 34113, Republic of Korea}

\definecolor{rkrPurple}{HTML}{73024F}

\begin{document}

\title{
    Compact localized currents in flat bands with broken time-reversal symmetry
}

\author{Rohit Kishan Ray\,\orcidlink{0000-0002-5443-4782}}
    \email{rkray@vt.edu}
    \affiliation{\pcsadd}
    \affiliation{
        Department of Material Science and Engineering, Virginia Tech, Blacksburg, VA 24061, USA
    }

\author{Carlo Danieli\,\orcidlink{}}
    \email{carlo.danieli@cnr.it}
    \affiliation{
        Istituto dei Sistemi Complessi, Consiglio Nazionale delle Ricerche, Via dei Taurini 19, 00185 Rome, Italy
    }

\author{Alexei Andreanov\,\orcidlink{0000-0002-3033-0452}}
    \altaffiliation{Current address: Center for Trapped Ions Quantum Science, Institute for Basic Science, Daejeon 34126, Republic of Korea}
    \email{aalexei@ibs.re.kr}
    \affiliation{\pcsadd}
    \affiliation{\ustadd}

\author{Sergej Flach\,\orcidlink{}}
    \email{sflach@ibs.re.kr}
    \affiliation{\pcsadd}
    \affiliation{\ustadd}
    \affiliation{
        Centre for Theoretical Chemistry and Physics, The New Zealand Institute for Advanced Study (NZIAS),
        Massey University Albany, Auckland 0745, New Zealand
    }

\date{\today}

\begin{abstract}
    We develop a systematic framework for constructing all-bands-flat (ABF) lattice Hamiltonians that explicitly break time-reversal symmetry (TRS).
    By threading magnetic flux through disconnected polygonal plaquettes and applying local entangling unitary transformations, we map plaquettes onto families of ABF models in one, two, and three dimensions.
    This procedure preserves the flux configuration while converting semi-detangled geometries into ABF lattices with nontrivial hopping structure.
    The resulting flat bands admit compact localized states (CLSs) whose support includes both the flux-threaded plaquettes and auxiliary sites introduced by the unitary transformations.
    In these TRS-broken constructions, the CLSs host localized circulatory currents whose magnitude depends on the applied flux.
    We further extend the framework to lattices with coexisting flat and dispersive bands, illustrating cases with both orthogonal and non-orthogonal CLSs.
    Our results provide a controlled route for generating dispersionless lattices supporting flux-induced local currents.
\end{abstract}
\maketitle

\section{Introduction}
\label{sec:intro} 

Flat bands (FBs), {\it i.e.}, completely dispersionless energy spectra across the entire Brillouin zone, represent a remarkable class of electronic and photonic structures in condensed-matter physics and optics~\cite{derzhko2015strongly,leykam2018artificial,leykam2018perspective,vicencio2021photonic,rhim2021singular,regnault2022topological,leykam_2024_flat,danieli2024flat}. 
This absence of dispersion in lattice systems with short-range hopping causes the FB eigenstates to be non-zero only within a finite region due to destructive wave interference~\cite{read2017compactly}.
Hence, these states are referred to as compact localized states (CLSs). 
Interest in flat-band networks stems from their extreme sensitivity to perturbations, which gives rise to unconventional phenomena and a variety of nontrivial phases. 
Furthermore, FBs have been experimentally realized in a number of setups, from photonic waveguides~\cite{vicencio2015observation,mukherjee2015observation,weimann2016transport} to ultra-cold atoms~\cite{taie2015coherent,he2021flat}, electrical circuits~\cite{chase2024compact,lape2025realization} and acoustic~\cite{ma2021acoustic,shen2022observing} meta-materials  (see Ref.~\cite{danieli_2026_progress} for a review of recent progress in artificial FBs). 


Since the inception of flat-bands~\cite{sutherland1986localization,mielke1991ferromagnetism,tasaki1992ferromagnetism}, one of the most active areas of research has been the generation and the classification of FB lattices. 
Generating protocols developed over the years include methods based on chiral symmetry~\cite{mielke1991ferromagnetism}, including the presence of magnetic flux~\cite{aoki1996hofstadter}, origami schemes~\cite{dias2015origami}, repetition of mini-arrays~\cite{morales-inostroza_2016_simple}, local symmetries~\cite{rontgen_2018_compact},
flat bands on fractal geometries~\cite{pal2018flat,biswas_2023_designer}, and construction from interacting frustrated clusters~\cite{santos_2020_methods}, among many others~\cite{hwang2021general,calugaru2022general,ryu2024orthogonal}. 
A more systematic approach emerged from realizing that FBs can be generated by assuming a specific CLS profile and fine-tuning the Hamiltonian networks that support them. This CLS-based approach yielded systematic flat-band generators in various dimensions~\cite{maimaiti2017compact,maimaiti2019universal,maimaiti2021flat,chen_2023_decoding} and an exhaustive classification of FBs.

FBs indeed can be classified according to the orthogonality and completeness of the CLS set. For generic flatbands, CLSs form a linearly independent set of non-orthogonal states---corresponding to \emph{non-orthogonal/linearly independent flatband}~\cite{kim2026real}.
In certain cases, however, the CLS set is 
(i) orthonormal~\cite{flach2014detangling}---giving an \emph{orthogonal flatband}; or 
(ii) alinearly dependent one. 
The latter class, which exists only for \(D\geq 2\) spatial dimensions is also known as a \emph{singular flat band}~\cite{rhim2019classification,rhim2021singular,graf2021designing,kim2023general}.
In this case, the linear dependence of the CLS enforces that one the remaining dispersive bands touches the flat band.

A notable case of orthonormal FBs is when a lattice completely lacks dispersion---\emph{i.e.}, all spectral Bloch bands are flat.
Such systems are known as all-bands-flat (ABF) lattices. 
This scenario was introduced by Vidal \emph{et.al.}, demonstrating that certain lattices threaded with fine-tuned magnetic flux feature exclusively flat bands~\cite{vidal1998aharonov,vidal2000interaction}.
In this case, since all eigenstates are compact, any initially localized state will remain confined within a finite region---a phenomenon known as caging, which leads to a complete suppression of transport.
The elimination of the single-particle transport mechanisms offers a fertile platform for investigating strongly correlated phases in the interacting case, from nonlinear caging~\cite{gligoric2019nonlinear,diliberto2019nonlinear,danieli2021nonlinear} to 
many-body CLSs~\cite{tovmasyan2018preformed,tilleke2020nearest,danieli2021quantum},   
quantum scars~\cite{hart2020compact,kuno2020flat_qs,kuno2021multiple,pelegri2024few}, and exact many-body flat band localization~\cite{danieli2020many,kuno2020flat22,orito2021nonthermalized}.
ABFs have also been experimentally realized with photonic waveguides~\cite{mukherjee2018experimental,kremer2020square,Jorg2020artificial,caceres2022controlled,roman2025observation,vicencio2025multi}, photonics~\cite{xia2025fully}, Fock space~\cite{yang2023realization}, cold atoms~\cite{kang2020creutz,li2022aharonov}, acoustic metamaterials~\cite{samak2024direct}, electric~\cite{wang2022observation, lape2025realization} and superconducting~\cite{martinez2023interaction} circuits. 

We focus our attention on ABF lattices where time-reversal symmetry (TRS) is explicitly broken by magnetic flux. 
This focus is motivated by experiments where intrinsic magnetic fields were used to break the TRS in systems featuring nearly flat bands, leading to anomalous quantum Hall effect~\cite{han_2024_Largequantum, lu_2025_Extendedquantum}. 
Furthermore, it has been demonstrated that introducing nonuniform fluxes produce flat bands~\cite{katsura2010ferromagnetism} or multiple fluxes can fine-tune photonic lattices into ABF regime while breaking TRS~\cite{brosco2021two}.
These developments raise several fundamental questions that motivate the present work:
(i) can ABF lattices with broken TRS be systematically constructed,
(ii) can this construction be extended to lattices supporting both flat and dispersive bands;
and 
(iii) does the broken TRS induce localized currents in the CLS. 
In this work, we address these questions by developing a systematic scheme for constructing ABF models with broken TRS.
The scheme builds on the local detangling properties characteristic of ABF lattices~\cite{danieli2021nonlinear}.
We construct explicit parametric families of such lattices in one, two, and three dimensions, and show that their CLSs support localized circulating currents whose magnitude depends on the applied magnetic flux.
We further extend the framework to generate systems featuring coexisting flat and dispersive bands, illustrating examples with both orthogonal and non-orthogonal CLSs.

This paper is organized as follows.
In Section~\ref{sec:TRS}, we present the theoretical background.
Section~\ref{sec:construction} details our construction method for engineering TRS-broken ABF models in 1D, 2D, and 3D, and demonstrates the resulting localized currents.
In Section~\ref{sec:construction_FB}, we extend this discussion to lattices with coexisting flat and dispersive bands. 
Finally, Section~\ref{sec:discussion} summarizes our findings and presents the conclusions.

\section{Time reversal symmetry in Hamiltonian lattices}
\label{sec:TRS}

Let us consider a \(D\) dimensional lattice with \(\nu\) bands (\textit{i.e.}, \(\nu\) sites per unit-cell)
described by an Hermitian Hamiltonian
\begin{align}
    \label{eq:H}
    \begin{split}
        H  &= \sum_{{\bf m}} \sum_{\mu,\eta}\, [H_0]_{\mu,\eta}\,\ketbra{{\bf m},\mu}{{\bf m},\eta}\, \\
        &+\sum_j \sum_{\langle {\bf m},{\bf m}' \rangle_{s,j}}\sum_{\mu,\eta}\, [H_j^s]_{\mu,\eta}\,\ketbra{{\bf m},\mu}{ {\bf m}',\eta} + \mathrm{h.c.}
    \end{split}
\end{align} 
The basis states $\ket{{\bf m},\mu} = a_{{\bf m},\mu}^\dagger\ket{\varnothing}$ are defined through the creation and annihilation operators $a_{{\bf m},\mu}^\dagger$ and $a_{{\bf m},\mu}$.
The multi-index ${\bf m}\in\mathbb{Z}^D$ locates a unit cell within the lattice and the index $1\leq \mu\leq \nu$ labels sites within each cell. 
The wave function is then expressed as $\ket{\Psi} = \sum_{{\bf m},\mu } \psi_{{\bf m},\mu} \ket{{\bf m},\mu}$. 
The Bravais lattice is defined by the primitive lattice translation vectors $\Vec{a}_j$ for $1\leq j\leq D$. 
The square matrix $H_0$ of size $\nu$ defines the unit cell profile, while the square matrices $H_j^s$ define the hopping between $s^\mathrm{th}$ order-neighboring unit cells $\langle {\bf m},{\bf m}' \rangle_{s,j}$ at distance $|{\bf m}-{\bf m}'|=s$ along the vectors $\Vec{a}_j$.
Here we consider only finite range hopping between unit cells, $s<+\infty$, and finite number of Bloch bands, $\nu<+\infty$.  
The eigenvalue problem of the Hamiltonian $H$ in Eq.~\eqref{eq:H} is $E\ket{\Phi}=H\ket{\Phi}$.
Using the Bloch space representation $\ket{\Psi} = e^{\imi {\bf k} \cdot {\bm{m}}} \ket{\Phi}$ for the wave vector ${\bf k}$ yields $\nu$ Bloch bands given by $\{E_\mu({\bf k})\}_{\mu=1}^\nu$.

Let us now introduce the concept of time-reversal symmetry for Hamiltonian lattices.
An Hamiltonian \(H\) acting on a Hilbert space $\Hil$ is time-reversal symmetric (TRS) if there exist an anti-unitary time-reversal operator \(\mathcal{T}\) such that~\cite{wigner_2012_group,cordova_2018_timereversal} 
\begin{gather}
    \label{eq:trs_def}
    \mathcal{T}H\mathcal{T}^{{-}1} = H. 
\end{gather}
The anti-unitary time-reversal operator $\mathcal{T}$ can be written as a composition of two operators as --- \(\mathcal{T} = \mathcal{U} \mathcal{K} \) where \( \mathcal{U}\) is a unitary rotation and \( \mathcal{K}\) is the standard complex conjugation.
In case of spinless Hamiltonians (the focus of our work) we have \( \mathcal{U} = \mathbb{I}\) and the time-reversal operator reduces to \(\mathcal{T} = \mathcal{K} \).
In our context, TRS yields a \({\bf k}\to {-}{\bf k}\) symmetry in the Bloch energy bands, which is typically broken in the absence of TRS.

Importantly, time-reversal symmetry is not generically preserved under arbitrary unitary transformations.
When a Hamiltonian is transformed as $H^\prime = \bmu H \bmu^\dagger$, a necessary and sufficient condition for TRS preservation is that the operator $\mathcal{W}:= \bmu^\dagger\mathcal{T}\bmu\mathcal{T}^{{-}1}$ commutes with $H$(see Appendix~\ref{appA:TRS}).
In the case of spinless Hamiltonians, it can be shown that the above condition reduces to $\bmu = \bmu^\ast$. 
Therefore, for a unitary drawn at random from the full unitary group \(\mathrm{U}(n)\), TRS is almost surely broken unless:
(i) we specifically restrict $\bmu$ to be a real (orthogonal) matrix, or
(ii) $\bmu$ is such that the transformed Hamiltonian $H^\prime = H^{\prime\ast}$~\footnote{we denote the complex conjugate of a matrix $\bm{A}$ by $\bm{A}^\ast$} and $\bmu^\dagger\mathcal{T}\bmu\mathcal{T}^{{-}1}$ commutes with $H$~\cite{heinzner_2005_symmetry} (see Appendix~\ref{appA:TRS}).
%

\section{All-bands-flat lattices with broken time reversal symmetry}
\label{sec:construction}

In this section, we focus on constructing Hamiltonian lattices $H$ in Eq.~\eqref{eq:H} where each of the $\nu$ Bloch bands is independent of momentum ${\bf k}$ over the entire Brillouin zone, $E_\mu({\bf k}) = \mathrm{const}$ for $1\leq \mu\leq \nu$.
This class of lattice networks lacking any band dispersion is called all-band-flat (ABF).
This complete flattening of the spectral bands is highly nontrivial, as all eigenstates are macroscopically degenerate and are spatially compact, hence preventing any single particle transport along the lattice, resulting in caging.
The flat band eigenstates are compact localized states (CLSs).

Following Ref.~\cite{danieli2021nonlinear}, a representation of an ABF lattice with Hamiltonian $H$ in Eq.~\eqref{eq:H} featuring both non-zero intra-cell matrix $H_0$ and inter-cell hopping matrices $H_j^s$ is called \emph{Non-Detangled} (ND)~\cite{danieli2020many}. 
In this case the values of the hopping profile encoded in the matrices $H_0$ and $H_j^s$ are fine-tuned to impose destructive interference---hence, resulting in the formation of compact localized states and the emergence of flat bands.
The CLS of an ABF lattice form an orthogonal and complete set of the Hilbert space~\cite{danieli2021nonlinear,sathe2025topological}.
The orthogonality and completeness of the CLS imply that there exist specific choices of the unit-cell representation that eliminate all the inter-cell hopping matrices, $H_j^s = 0$.
This type of representations is called \emph{Semi-Detangled} (SD), and the lattice in then formed by decoupled unit cells~\footnote{
The choice of unit cell representation can be further fine-tuned in order to diagonalize $H_0$ in Eq.~\eqref{eq:H} while keeping $H_j^s=0$.
In this case, the ABF lattices are called Fully-Detangled.}. 
An all-band-flat lattice therefore can be recast from the ND representation into the SD representation (and \emph{vice versa}) via local unitary operators $\bmu$---\textit{i.e.,} operating between neighboring unit-cells~\cite{danieli2021nonlinear}. 

We now use this entangling process via local unitary operator to construct ABF lattices with broken TRS.
The construction scheme begins from the SD representation of an ABF. 
First, we choose plaquettes (\emph{i.e.,} a closed loop of $M$ sites that does not intersect itself) threaded with a magnetic field as detangled unit cells---hence breaking TRS.
We switch to an ND representation through the rotations $\bmu$ that do not alter TRS status of the lattice.
The rotation $\bmu$ is formed of local unitary transformations $U\!\in\!\gsun{r}$ involving $r$ sites of two or more neighboring plaquettes.   
Such procedure however, in general produces highly intricate hoppings which may have limited physical applicability, although mathematically sound. 
To counter this and generate the ABFs without deforming the plaquettes with threaded magnetic fields, we attach auxiliary sites to the plaquettes (a.k.a appendices) and apply the unitaries $\bmu$ to entangle those neighboring appendices.

Two comments are in order:
First, the ABF models we consider may admit an interpretation via the Carrollian symmetry~\cite{ara2025flat}, although the full formal proof of equivalence is missing at the moment.
Second, linearly independent and singular flatbands do not by definition admit the above construction:
because their CLS sets are not orthonormal such flatbands cannot be diagonalized or constructed by local unitary transformations~\cite{kim2026real}.

We illustrate our scheme with a one dimensional lattice with the smallest possible number of Bloch bands.
We progress subsequently to more complex cases. 

The cases presented below are succinctly formulated in a systematic algorithm~\ref{algo:construction}, that applies to 1D and higher dimensions.
Table~\ref{tab:summary} summarizes the models considered in this article; outlining the dimensions, plaquette shape, CLS type, number of bands, and types of inter-cell hopping.

Our work builds on previous results on constructing ABF lattices~\cite{danieli2021nonlinear} (and FB lattices in general) and generalizes previous work on FBs with magnetic fluxes to the level of systematic construction, notably that of Vidal~\cite{vidal1998aharonov,vidal2000interaction}.

\begin{algorithm}[t!]
\caption{Construction of ABF lattices with broken TRS}
\label{algo:construction}
\KwIn{Set of disconnected polygons (plaquettes) of \(M\) sites (SD lattice) threaded by a flux $\phi$ (TRS breaking)}
\KwOut{A connected ABF network with $\Tilde{M}=M+M^\prime$ bands, inter-cell hoppings connecting nearest-neighbor sites}

\textbf{Provided:}
\begin{enumerate}
    \item \(M^\prime\) auxiliary sites that branch out,
    \item The coordination number $r$---number of nearest-neighbor auxiliary sites to be entangled,
    \item Local TRS preserving unitary operation $\bmu$ to produce \(\mathrm{SU}(r)\) rotations between the neighboring auxiliaries.
\end{enumerate}
\textbf{Steps:}\\
\begin{enumerate}
    \item Begin with the intra-cell matrix $H_0$ of dimension $\Tilde{M}\times\Tilde{M}$ (plaquette and appendices).
    \item Create a block-diagonal array $\bm{H_0}$ of $L_1 \times L_2 \times \ldots$ copies of $H_0$ (dimension $\Tilde{M}(L_1 \times L_2 \times \ldots)$).
    \item Select the auxiliary (appendices) sites to be entangled across neighboring unit cells.
    \item Construct the entangling transformation $\bmu$ using the methods described in Appendix~\ref{appA:unitary}.
    \item Apply $\bmu$ to $\bm{H_0}$ to get the ND Hamiltonian $H^\prime = \bmu \bm{H_0}\bmu^\dagger$.
    \item Compute the current in the plaquettes due to the eigenvectors using Eq.~\eqref{eq:prob_current}.
\end{enumerate}
\end{algorithm}

\begin{table*}[t]
    \centering
    \caption{
        Summary of all ABF and flat-band constructions presented in this work. 
        NN: nearest-neighbor; NNN: next-nearest-neighbor; FB: flat band.
    }
    \label{tab:summary}
    \begin{ruledtabular}
    \begin{tabular}{lclcccc}
        Figure & Dimension & Plaquette & CLS type & Bands $\tilde{M}$ & Inter-cell hopping & Auxiliary sites \\
        \hline
        Fig.~\ref{fig:1D_ABF}(b) & 1D & Triangle      & ABF              & 3  & Up to NNN & No  \\
        Fig.~\ref{fig:1D_ABF}(d) & 1D & Triangle    & ABF              & 5  & NN only   & Yes \\
        Fig.~\ref{fig:2D_ABF}(a) & 2D & Triangle   & ABF              & 6  & NN only   & Yes \\
        Fig.~\ref{fig:2D_ABF}(b) & 2D & Triangle, 12-site unit cell  & ABF              & 12 & NN only   & Yes \\
        Fig.~\ref{fig:ABF}    & 2D & Square      & ABF              & 8  & NN only   & Yes \\
        Fig.~\ref{fig:3D_ABF}(a) & 3D & Tetrahedron, 8-site unit cell  & ABF            & 8  & NN only   & Yes \\
        Fig.~\ref{fig:3D_ABF}(b) & 3D & Tetrahedron, 16-site unit cell & ABF            & 16 & NN only   & Yes \\
        Fig.~\ref{fig:FB_ortho1}(b) & 1D & Triangle $+$ dispersive chain & Orthogonal FB   & 4  & NN only    & No \\
        Fig.~\ref{fig:FB_ortho1}(d) & 1D & Triangle $+$ dispersive chain & Orthogonal FB   & 5  & NN only   & Yes \\
        Fig.~\ref{fig:FB_non_ortho1}    & 1D & Rhombic (diamond chain)      & Non-orthogonal FB & 3  & NN only   & No  \\
    \end{tabular}
    \end{ruledtabular}
\end{table*}

\subsection{A simple case}
\label{sec:construction_ABF_1d}

Let us exemplify our construction scheme through the simplest case: a one-dimensional ABF lattice built from disjoint triangular plaquettes.
We break TRS in each plaquette by threading them with a magnetic field of flux \( \phi\)~\footnote{A triangular plaquette constitutes the minimal loop necessary to induce TRS breaking via magnetic field}.
The corresponding matrix $H_0$ in Eq.~\eqref{eq:H} reads
\begin{align}
    \label{eq:1d_trig}
    H_{0,\,\scalebox{0.7}{$\Delta $}} = 
    \begin{pmatrix}
        0 & e^{\frac{\imi \phi}{3}} & e^{-\frac{\imi \phi}{3}} \\
        e^{-\frac{\imi \phi}{3}} & 0 & e^{\frac{\imi \phi}{3}} \\
        e^{\frac{\imi \phi}{3}} & e^{-\frac{\imi \phi}{3}} & 0
    \end{pmatrix}
\end{align}
This plaquette forms the unit-cell of the lattice Hamiltonian \(H\) of Eq.~\eqref{eq:H} and the resulting structure corresponds to a semi-detangled (SD) ABF lattice, Fig.~\ref{fig:1D_ABF}(a).
A triangular plaquette in Fig.~\ref{fig:1D_ABF}(a) has $\nu=3$ flat bands (eigenstates of \(H_0\)~\eqref{eq:1d_trig})
\begin{align}
    E_1 = 2 \cos\frac{\phi}{3}, \qquad E_{2,3} = -\cos\frac{\phi}{3} \pm \sqrt{3} \sin\frac{\phi}{3}
    \label{eq:ABF_tr}
\end{align}
These flat band eigenstates are completely localized within each plaquette yet carry a non-zero probability current due to the enclosed magnetic flux.
We can compute this probability current $\mathcal{J}_{\bm{m};\,\mu,\eta}$ induced by a wavefunction $\psi$ between two adjacent lattice sites $\mu$ and $\eta$ of neighboring unit-cell at $\bf{m}$ and $\bf{m^\prime}$, respectively, using~\cite{shankar_1994_PrinciplesQuantum}
\begin{gather}
    \label{eq:prob_current}
    \mathcal{J}_{\bm{m},\mu;\bm{m^\prime},\eta} = \imi \left\{\psi^\ast _{\bm{m},\,\mu} [\mathit{H}]_{\mu,\eta} \psi_{\bm{m^\prime},\,\eta} - \psi_{\bm{m^\prime},\,\eta}^\ast [\mathit{H}^\dagger]_{\eta,\mu} \psi_{\bm{m},\,\mu} \right\},
\end{gather}
where $\psi_{\bm{m},\,\nu}$ is the $\nu^\Th$ site component of the wavefunction of unit-cell at $\bf{m}$.
Here $H$ is the intra-cell hopping matrix $H_0$ if $\bm{m}=\bm{m^\prime}$, else it is the inter-cell hopping matrix $H_j^s$ between neighboring unit-cells at $\bm{m}$ and $\bm{m^\prime}$ in the $\Vec{a}_j$ direction.
For instance, the eigenstates associated to $E_1=2\cos \frac{\phi}{3}$ has equal amplitudes in all three sites $\psi_{m,\mu}=\frac{1}{\sqrt{3}}$ for $\mu=1,2,3$, which implies $\mathcal{J}_{m,\mu,m,\eta} = \frac{2}{3}\sin \frac{\phi}{3}$ in each bond. 
The total current in a flux-threaded plaquette $m$ corresponding to a CLS is computed via the sum $\mathcal{J} = \sum_\mu \mathcal{J}_{m,\mu,m,\mu+1}$ using Eq.~\eqref{eq:prob_current}. 
In the lattice in Fig.~\ref{fig:1D_ABF}(a), 
the total current corresponding to the CLS associated to $E_1=2\cos \frac{\phi}{3}$ is then $\mathcal{J} =2\sin \frac{\phi}{3}$. 

\begin{figure}
    \centering
    \includegraphics[width=\columnwidth]{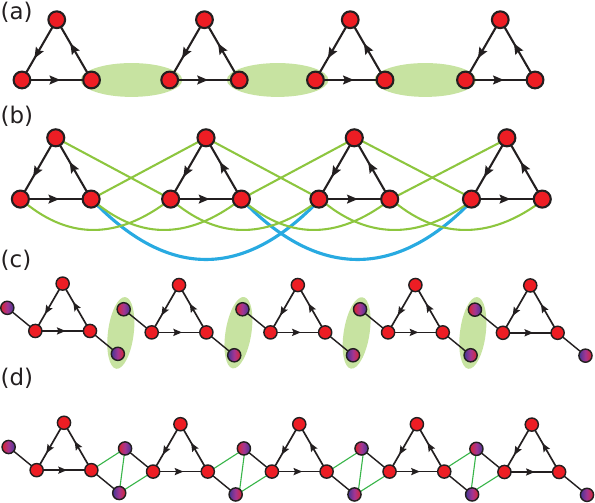}
    \caption{
        (a) 1D lattice formed by stacking unit cells, that are triangles threaded by a flux \( \phi\).
        Green shadings mark the sites to which local unitary operation is applied
        (b) Resulting lattice after the entangling rotations.
        (c),~(d) same as (a),~(b) but for triangular plaquettes with auxiliary sites. 
    }
    \label{fig:1D_ABF}
\end{figure}

To transition from SD to ND representation without breaking the TRS nature of the Hamiltonian $H$, (see Appendix~\ref{appA:TRS} for details), we design a suitable unitary transformation $\bmu$ formed of rotations $U \in \gsun{2}$ 
\begin{gather}
    \label{eq:SO2_U}
    U = 
    \begin{pmatrix}
        \cos\theta & \sin\theta \\
        -\sin\theta & \cos\theta
    \end{pmatrix}
\end{gather} 
for an angle $\theta\in[0,2\pi]$.
This transformation $\bmu$ acts on pairs of sites from neighboring triangular plaquettes --- as shown in Fig.~\ref{fig:1D_ABF}(a) with green shading (see Appendix~\ref{appA:unitary_su2} for details on the construction of $\bmu$).  
The resulting ND lattice is shown in Fig.~\ref{fig:1D_ABF}(b).
The intra-cell matrix of its Hamiltonian $H^\prime$ is 
\begin{align}
    H^\prime_{0,\,\scalebox{0.7}{$\Delta $}} = 
    \begin{pmatrix}
        0 & e^{\frac{\imi\phi}{3}} \cos \theta  & e^{-\frac{\imi\phi}{3}} \cos \theta  \\
        e^{-\frac{\imi\phi}{3}} \cos \theta  & 0 & e^{\frac{\imi\phi}{3}} \cos \theta  \\
        e^{\frac{\imi\phi}{3}} \cos \theta  & e^{-\frac{\imi\phi}{3}} \cos \theta  & 0
    \end{pmatrix}
\end{align}
while the nearest-neighbor hopping matrix is
\begin{align}
    H^{1\,\prime}_{1,\,\scalebox{0.7}{$\Delta $}} = 
    \begin{pmatrix}
        0 & 0 & -e^{\frac{\imi \phi}{3}} \sin \theta  \\
        e^{\frac{\imi \phi}{3}} \sin \theta  & \half e^{-\frac{\imi \phi}{3}} \sin 2\theta  & 0 \\
        0 & 0 & -e^{-\frac{\imi \phi}{3}} \sin \theta
    \end{pmatrix}
\end{align}
Additionally, this lattice has next-to-nearest neighbor hopping terms, which are encoded in the matrix 
\begin{gather}
    H^{2\,\prime}_{1,\,\scalebox{0.7}{$\Delta $}} = 
    \begin{pmatrix}
        0 & 0 & 0 \\
        0 & 0 & -e^{-\frac{\imi\phi}{3}} \sin ^2\theta  \\
        0 & 0 & 0
    \end{pmatrix}
\end{gather}
In Fig.~\ref{fig:1D_ABF}(b), the hopping terms encoded by these three matrices are denoted by bonds colored in black, green and blue, respectively.

This construction does produce a family of ABF lattices with broken TRS.
However, the resulting inter-cell hoppings in general have complex topology (Fig.~\ref{fig:1D_ABF}(b)) deviating from the  original objective of simple and intuitive lattices.
To counter this, we embed each triangular plaquette in Fig.~\ref{fig:1D_ABF}(a) in an extended unit cell by appending two auxiliary sites, as shown in Fig.~\ref{fig:1D_ABF}(c).
This extends the unit cell to five sites, whose intra-cell matrix $H_0$ reads  
\begin{gather}
     \label{eq:1d_trig_2}
    H_{0,\scalebox{0.6}{\tikztriangle}} = 
    \begin{pmatrix}
        0 & e^{\frac{\imi \phi}{3}} & e^{-\frac{\imi \phi}{3}} & 1 & 0 \\
        e^{-\frac{\imi \phi}{3}} & 0 & e^{\frac{\imi \phi}{3}} & 0 & 1  \\
        e^{\frac{\imi \phi}{3}} & e^{-\frac{\imi \phi}{3}} & 0 & 0 & 0 \\
        1 & 0 & 0 & 1 & 0 \\
        0 & 1 & 0 & 0 & 1
    \end{pmatrix}
\end{gather} 
Similar to the case of the triangular plaquette in Eq.~\eqref{eq:1d_trig}, the five CLSs posses non-zero currents within the plaquette, which can be computed via  Eq.~\eqref{eq:prob_current}. 

A unitary transformation $\bmu$ that entangles this SD lattice, in this case, involves the auxiliary sites of two adjacent plaquettes---as shown with green shaded areas of Fig.~\ref{fig:1D_ABF}(c). 

The corresponding intra- and inter-cell matrices $H_0$ and $H_1$ of the ND Hamiltonian become
\begin{gather}
\label{eq:H0Hx_1d_2}
\begin{split}
     H^\prime_{0,\, \scalebox{0.6}{\tikztriangle}}
     &= 
    \begin{pmatrix}
        0 & e^{\frac{\imi \phi}{3}} & e^{-\frac{\imi \phi}{3}} & \cos \omega & 0 \\
        e^{-\frac{\imi \phi}{3}} & 0 & e^{\frac{\imi \phi}{3}} & 0 & \cos \omega \\
        e^{\frac{\imi \phi}{3}} & e^{-\frac{\imi \phi}{3}} & 0 & 0 & 0 \\
        \cos \omega  & 0 & 0 & 1 & 0 \\
        0 & \cos \omega & 0 & 0 & 1
    \end{pmatrix} \\
     H^\prime_{1,\scalebox{0.6}{\tikztriangle}} &= 
    \begin{pmatrix}
        0 & 0 & 0 & 0 & -\sin \omega  \\
        0 & 0 & 0 & 0 & 0 \\
        0 & 0 & 0 & 0 & 0 \\
        0 & \sin \omega  & 0 & 0 & 0 \\
        0 & 0 & 0 & 0 & 0
    \end{pmatrix}
\end{split}
\end{gather}
With respect to the previous ND lattice shown in Fig.~\ref{fig:1D_ABF}(b), the hopping terms that connect neighboring plaquettes are simplified to rhombii. 
Intuitively, these appendices induce ``directions" along which the disconnected plaquettes are entangled, forming a lattice.
This simplification of the hopping structure in the ND representation comes at the cost of the number of bands increasing from three to five, and subsequent increase of the dimensionality of Hilbert space.
Nevertheless, the entangling process does not perturb the triangular plaquettes, nor the CLS components sitting therein.
Hence, the local currents $\mathcal{J}_{\mu,\eta}$ are untouched.
These CLSs, which exhibit localized probability current due to magnetic flux within the plaquette act as \emph{vortices} in the lattice.
Consequently, TRS breaking induces a macroscopically degenerate family of linear vortex states.



\subsection{Two and three dimensional lattices}
\label{sec:kritzel_kratzel} 

\begin{figure}
    \centering
    \includegraphics[width=\columnwidth]{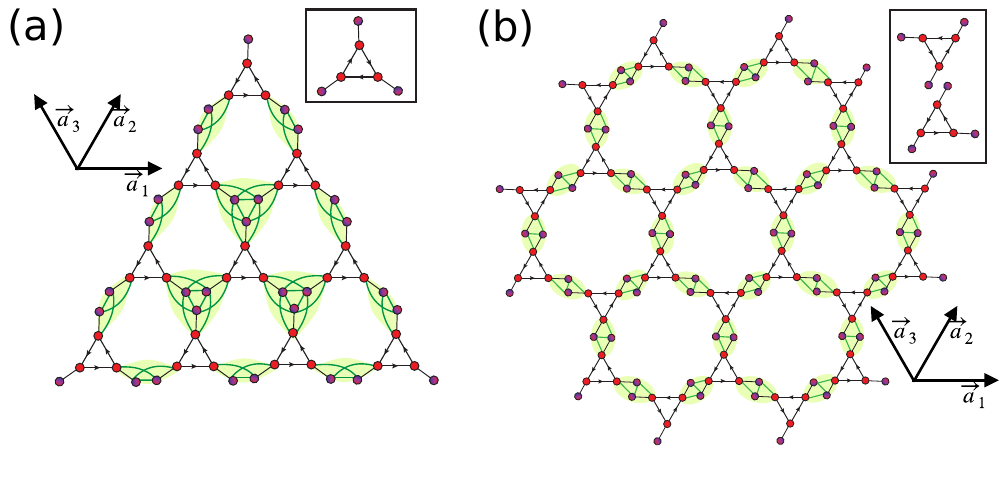}
    \caption{
        (a) 2D ABF lattice with a different tiling of triangular unit cells (inset).
        In both cases hoppings induced by local unitary transformation are shown by green bonds.    
        (b) 2D ABF lattice generated from hexagonal tiling of triangular unit cells threaded by a flux \( \phi\) (inset). 
        In both panels, $\Vec{a}_1$ and $\protect\Vec{a}_2$ are the primitive translation vectors, with $\protect\Vec{a}_3 = \protect\Vec{a}_2 - \protect\Vec{a}_1$. The green shaded regions denote the sites affected by the entangling unitary.
}
    \label{fig:2D_ABF}
\end{figure}

In one dimension, the construction scheme shown in Fig.~\ref{fig:1D_ABF} which starts from flux-threaded triangular plaquettes can be straightforwardly extended to larger plaquettes, \textit{e.g.,} squares, pentagons, and beyond.
The choices of the construction of the entangling unitary transformations $\bmu$ and the placement of auxiliary sites create a wide variety of ABF lattices with broken time-reversal symmetry. 
In the previous section, our results also highlighted that reducing the complexity of the hopping profile of the lattice can be achieved by increasing the number of sites per unit cell (\textit{i.e.}, the number of flat bands).
We now extend this construction scheme to higher dimensional lattices, where an additional choice in the construction is the arrangement of the primitive lattice translation vectors $\Vec{a}_j$ in Eq.~\eqref{eq:H}, \textit{i.e.}, the choice of the Bravais lattice. 

To describe the construction in higher dimensions, as in the previous subsection, we use triangular plaquettes. 
In two dimensions, the basic unit cell is shown in the inset of Fig.~\ref{fig:2D_ABF}(a), and it is formed by a flux-threaded three-sites loop with an appendix site attached at each corner.
Its Hamiltonian matrix reads
\begin{gather}
    H_{0,\scalebox{0.6}{\tikztriangle}} = 
    \begin{pmatrix}
        0 & e^{\frac{\imi \phi}{3}} & e^{-\frac{\imi \phi}{3}} & 1 & 0 & 0\\
        e^{-\frac{\imi \phi}{3}} & 0 & e^{\frac{\imi \phi}{3}} & 0 & 1 & 0 \\
        e^{\frac{\imi \phi}{3}} & e^{-\frac{\imi \phi}{3}} & 0 & 0 & 0 & 1\\
        1 & 0 & 0 & 1 & 0 & 0\\
        0 & 1 & 0 & 0 & 1 & 0\\
        0 & 0 & 1 & 0 & 0 & 1
    \end{pmatrix}
     \label{eq:1d_trig_3}    
\end{gather}
These six-sites unit cells are periodically repeated along the primitive vectors $\Vec{a}_1$ and $\Vec{a}_2$ shown in Fig.~\ref{fig:2D_ABF}(a), which form a $\frac{\pi}{3}$ angle.
The SD lattice is then entangled by applying local unitary rotations from \(\gsun{3}\) involving groups of three neighboring auxiliary sites.
This entangling procedure result in the green bonds shown in Fig.~\ref{fig:2D_ABF}(a) that link the neighboring flux-threaded triangular plaquettes without altering the TRS breaking of the system, see Appendix~\ref{appA:unitary_su3} for details. 
The resulting hopping network that links neighboring triangular plaquettes is again rather complex.
To reduce the complexity and design physical lattices, we extend the six site unit cell to a twelve site unit  cell by pairing two six site structures in Eq.~\eqref{eq:1d_trig_3} one on top of the other; as shown in the inset of Fig.~\ref{fig:2D_ABF}(b). 
These extended unit cells are then periodically arranged along the same primitive vectors as in Fig.~\ref{fig:2D_ABF}(a) so that only pairs of auxiliary sites will be entangled via \(\gsun{2}\) unitary rotations in Eq.~\eqref{eq:SO2_U}.
The bonds between the auxiliary sites created by entangling are shown in green.
This construction simplifies the hopping geometry while preserving the flux-threaded plaquettes; hence, resulting in a twelve-band ABF lattice whose geometry resembles the two-dimensional pyrochlore structure. 

\begin{figure}
    \centering
    \includegraphics[width=\columnwidth]{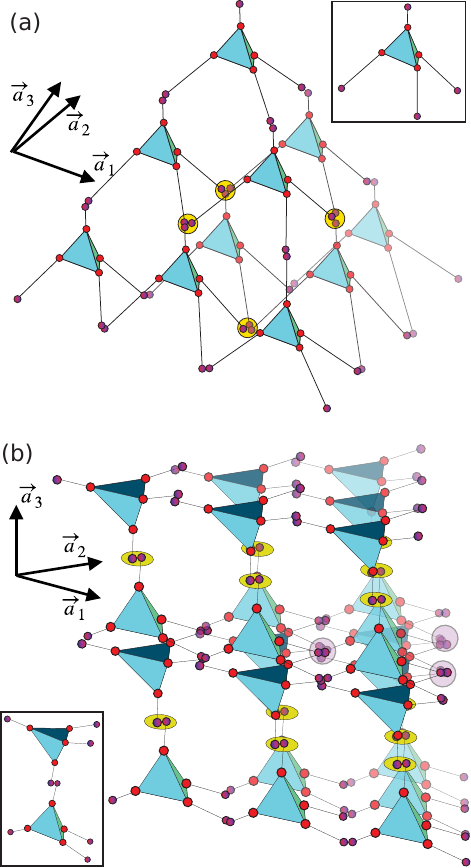}
    \caption{
        (a) 3D ABF lattice generated from pyramid-like tiling of tetrahedral unit cells with appendices (shown in the inset).
        (b) a different tiling producing a 3D ABF network with \(16\)-band and a tetrahedral unit cell with appendices (see inset).
        In both the panels, the plaquettes are denoted by pyramids shaded in blue, through which the flux is being threaded, $\Vec{a}_1,\Vec{a}_2$ and $\Vec{a}_3$ are the primitive translation vectors and the sites undergoing local unitary transformations are shown in shaded yellow and purple ellipses. Different shades of blue are used to  maintain the 3D perspective of the tetrahedrons.
    }
    \label{fig:3D_ABF}
\end{figure}

In three dimensions, we extend the triangular plaquettes to tetrahedra, where each triangular face is threaded by a magnetic flux which ensures TRS breaking.
To each vertex of each tetrahedron we append an auxiliary site.
This results in an eight-site unit cell shown in Fig.~\ref{fig:3D_ABF}(a) whose matrix extends Eqs.~(\ref{eq:1d_trig_2},\ref{eq:1d_trig_3}) to an eight-by-eight matrix. 
These unit cells are periodically arranged along the primitive lattice vectors $\Vec{a}_1$, $\Vec{a}_2$ and $\Vec{a}_3$, and entangled using local \(\gsun{4}\) transformations acting on groups of neighboring auxiliary sites (yellow shaded regions in Fig.~\ref{fig:3D_ABF}(a)).
The resulting eight-band ABF lattice has a complex hopping structure that links the neighboring tetrahedra.
This structure, obtained via a set of \(\gsun{4}\) unitary rotations, can again be simplified by increasing the number of bands. 
Consider the case of Fig.~\ref{fig:3D_ABF}(b): here we extend the unit cell to contain \(16\) bands, and therefore require a combination of \(\gsun{2}\) (yellow shaded regions) and \(\gsun{8}\) (purple shaded regions) rotations to entangle the neighboring unit cells. 
These two and three dimensional cases, shown in Fig.~\ref{fig:2D_ABF} and Fig.~\ref{fig:3D_ABF}, highlight that reduction in hopping complexity of a generated ABF typically comes at the cost of increasing the number of Bloch bands.

\subsection{Currents in compact localized states}
\label{sec:squares}

In this subsection, we focus on currents $\mathcal{J}$ which are induced on the CLS associated to the flat bands by the flux.
These currents are obtained by summing all the local probability currents $\mathcal{J}_{\bm{m},\mu;\bm{m^\prime},\eta}$ defined in Eq.~\eqref{eq:prob_current}.  
To discuss this, and to further present the construction of this class of ABF lattices, we consider the case of a flux-threaded square plaquette ($M=4$) with four appendices at each vertex.
This results in an eight-site unit-cell whose matrix $H_0$ in the Hamiltonian $H$ in Eq.~\eqref{eq:H} is
\begin{gather}
    \label{eq:H0_SD}
    H_{0,\scalebox{0.6}{\tikzsquare}} = 
    \left[ \begin{smallmatrix}
        0 & e^{\frac{\imi \phi}{4}} & 0 & e^{-\frac{\imi \phi}{4}} & 1 & 0 & 0 & 0 \\
        e^{-\frac{\imi \phi}{4}} & 0 & e^{\frac{\imi \phi}{4}} & 0 & 0 & 1 & 0 & 0 \\
        0 & e^{-\frac{\imi \phi}{4}} & 0 & e^{\frac{\imi \phi}{4}} & 0 & 0 & 1 & 0 \\
        e^{\frac{\imi \phi}{4}} & 0 & e^{-\frac{\imi \phi}{4}} & 0 & 0 & 0 & 0 & 1 \\
        1 & 0 & 0 & 0 & -1 & 0 & 0 & 0 \\
        0 & 1 & 0 & 0 & 0 & 1 & 0 & 0 \\
        0 & 0 & 1 & 0 & 0 & 0 & 1 & 0 \\
        0 & 0 & 0 & 1 & 0 & 0 & 0 & -1 \\
    \end{smallmatrix}\right].
\end{gather} 

\begin{figure}
    \centering
    \includegraphics[width=\columnwidth]{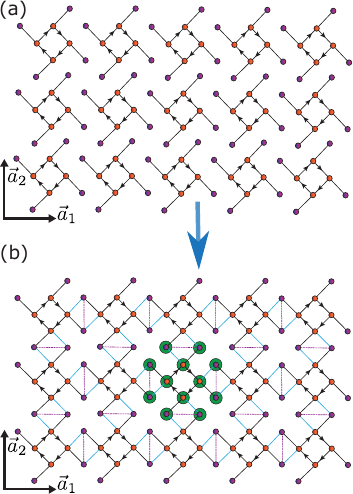}
    \caption{
        (a) Two dimensional semi-detangled ABF lattice obtained by tiling augmented square unit-cells.
        The squares, with the arrows on the links, are threaded by a flux.
        With the shaded areas we indicate the sites chosen for the local unitary operations $\bmu$.
        The green and the pink colors denote the $x$ and $y$ direction in which the two rotations $U_x$ and $U_y$ which form $\bmu$ are applied. 
        (b) Resulting non-detangled ABF lattice. 
        With the green circles we show the location of CLS. 
        In both the panels, $\Vec{a}_1$ and $\Vec{a}_2$ are the primitive translation vectors. 
    }
    \label{fig:ABF}
\end{figure}

We arrange these unit cells in two-dimensions following orthogonal primitive lattice vectors $\Vec{a}_1$ and $\Vec{a}_2$, which result in the SD lattice shown in Fig.~\ref{fig:ABF}(a). 
We entangle the auxiliary sites of the neighboring unit cells (highlighted area in pink and green in Fig.~\ref{fig:ABF}) via \(\gsun{2}\) unitary rotations $U_x$ and $U_y$ in the $x$ and $y$ directions.
In principle, these two local unitary rotations can be different.
For simplicity, considering the entangling rotation by $\omega$ in both direction,
the resulting intra-cell hopping matrix $H_{0,\scalebox{0.6}{\tikzsquare}}^\prime$ is
\begin{gather}
    \label{eq:h0_entangled}
    H_{0,\scalebox{0.6}{\tikzsquare}}^\prime= 
    \left[ \begin{smallmatrix}
        0 & e^{\frac{\imi\phi}{4}} & 0 & e^{-\frac{\imi\phi}{4}} & c & 0 & 0 & 0 \\
        e^{-\frac{\imi\phi}{4}} & 0 & e^{\frac{\imi\phi}{4}} & 0 & 0 & c & 0 & 0 \\
        0 & e^{-\frac{\imi\phi}{4}} & 0 & e^{\frac{\imi\phi}{4}} & 0 & 0 & c & 0\\
        e^{\frac{\imi\phi}{4}} & 0 & e^{-\frac{\imi\phi}{4}} & 0 & 0 & 0 & 0 & c\\
        c & 0 & 0 & 0 & -C_\text{\footnotesize $\omega$} & 0 & 0 & 0 \\
        0 & c & 0 & 0 & 0 & C_\text{\footnotesize $\omega$} & 0 & 0 \\
        0 & 0 & c & 0 & 0 & 0 & C_\text{\footnotesize $\omega$} & 0 \\
        0 & 0 & 0 & c & 0 & 0 & 0 & -C_\text{\footnotesize $\omega$}
    \end{smallmatrix}\right],
\end{gather}
while the inter-cell hopping matrices $H_{x,\scalebox{0.6}{\tikzsquare}}^\prime$, and $H_{y,\scalebox{0.6}{\tikzsquare}}^\prime$ are 
\begin{align}
    \label{eq:hxy_entangled}
    \begin{split}
        H_{x,\scalebox{0.6}{\tikzsquare}}^\prime \!=\! 
        \left[ \begin{smallmatrix}
            0 & 0 & 0 & 0 & 0 & 0 & -s & 0 \\
            0 & 0 & 0 & 0 & 0 & 0 & 0 & 0 \\
            0 & 0 & 0 & 0 & 0 & 0 & 0 & 0 \\
            0 & 0 & 0 & 0 & 0 & 0 & 0 & 0 \\
            0 & 0 & s & 0 & 0 & 0 & S_\text{\footnotesize $\omega$} & 0\\
            0 & 0 & 0 & 0 & 0 & 0 & 0 & 0 \\
            0 & 0 & 0 & 0 & 0 & 0 & 0 & 0 \\
            0 & 0 & 0 & 0 & 0 & 0 & 0 & 0 \\
        \end{smallmatrix}\right], &\,
        H_{y,\scalebox{0.6}{\tikzsquare}}^\prime \!=\! 
        \left[ \begin{smallmatrix}
            0 & 0 & 0 & 0 & 0 & 0 & 0 & 0  \\
            0 & 0 & 0 & 0 & 0 & 0 & 0 & -s  \\
            0 & 0 & 0 & 0 & 0 & 0 & 0 & 0 \\
            0 & 0 & 0 & 0 & 0 & 0 & 0 & 0 \\
            0 & 0 & 0 & 0 & 0 & 0 & 0 & 0 \\
            0 & 0 & 0 & s & 0 & 0 & 0 & -S_\text{\footnotesize $\omega$}\\
            0 & 0 & 0 & 0 & 0 & 0 & 0 & 0 \\
            0 & 0 & 0 & 0 & 0 & 0 & 0 & 0 \\
        \end{smallmatrix}\right]
    \end{split}
\end{align}
where $C_\text{\footnotesize $\omega$} = \cos(2\omega)$, $S_\text{\footnotesize $\omega$} \equiv \sin(2\omega)$,  $c = \cos(\omega)$ and $s = \sin(\omega)$.
The resulting ND ABF lattice is shown in Fig.~\ref{fig:ABF}(b) corresponding to a family of two dimensional eight-bands ABF lattices with broken TRS.
Each square plaquette which lies in the center of a unit cell is threaded with a magnetic field of flux $\phi$.
In Fig.~\ref{fig:ABF}(b) the green shaded regions show the sites of the eight compact localized states of the system. 

\begin{figure}
    \centering
    \includegraphics[width=\columnwidth]{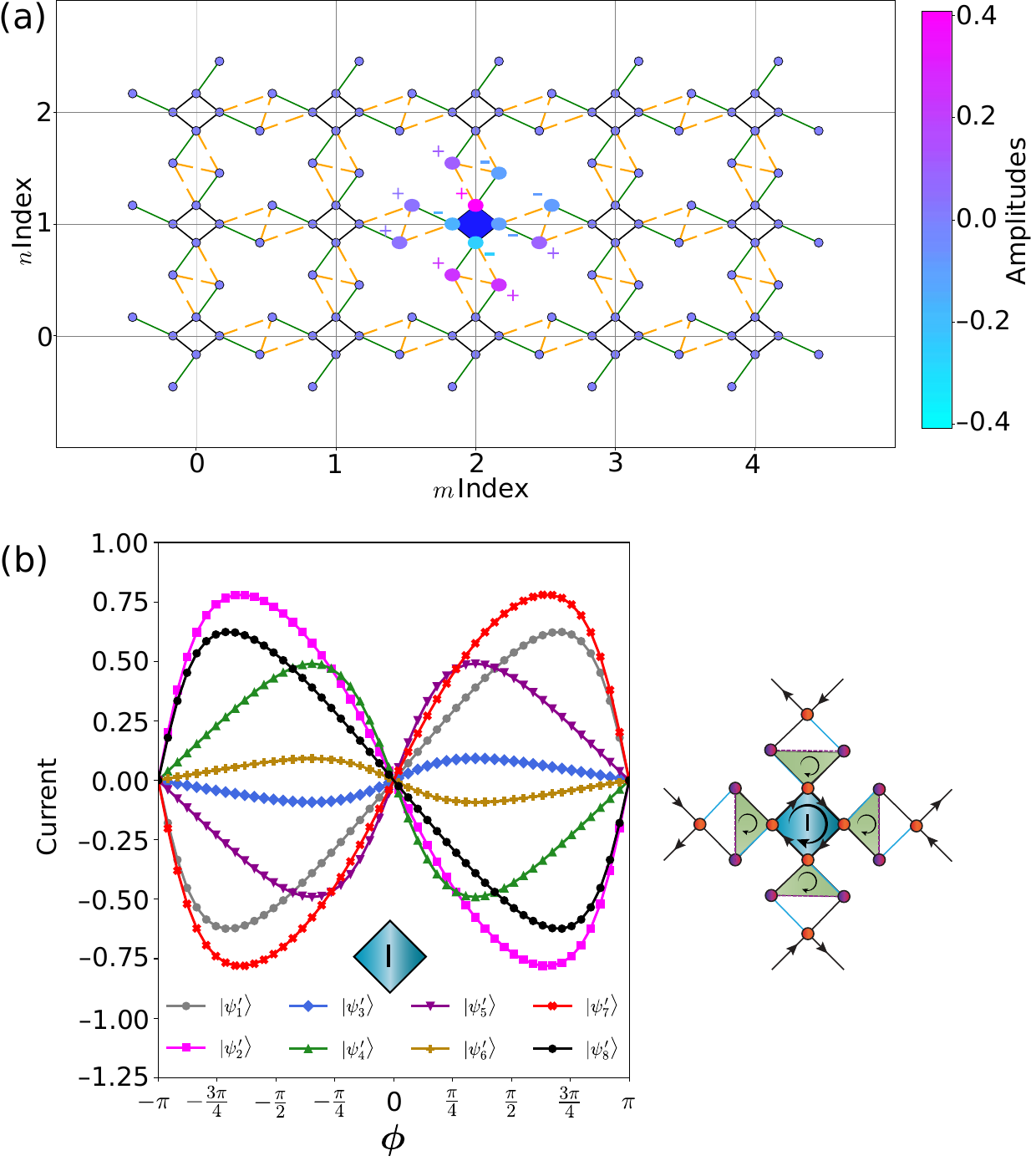}
    \caption{
        (a) CLS amplitude shown in bright colors, whose intensities are encoded in the right colorbar.
        The CLS sites are displayed as larger markers, and the signs of their amplitudes are explicitly indicated for clarity.
        The blue shaded region indicates where the current circulates. 
        (b) Probability current as a function of the flux $\phi$ for the eight different CLSs.
        On the bottom right, we indicate the plaquette labeled `${\sf I}$' where the current distribution is evaluated. There are zero currents in the green shaded triangles.
    }
    \label{fig:cls}
\end{figure}

In Fig.~\ref{fig:cls}(a), we show an example of one of the eight CLS at $E=-1.74$ of the non-detangle ABF lattice obtained for $\omega = \pi/4$ and $\phi=2$ in Eqs.~(\ref{eq:h0_entangled},\ref{eq:hxy_entangled}). 
The amplitudes of the CLS are shown with bright colors, the current in the plaquette is shown along with the associated colorbar.
The value of the current in the plaquette matches the corresponding value for the given parameters of Fig.~\ref{fig:cls}(b).
In Fig.~\ref{fig:cls}(b) we plot the probability currents $\mathcal{J}$ of all eight orthogonal CLSs as function of the flux $\phi$ for $\omega = \pi/4$.
These values are obtained by summing the local probability currents $\mathcal{J}_{\mu;\eta}$ defined in Eq.~\eqref{eq:prob_current} in the flux-threaded square plaquette highlighted in blue in the inset of Fig.~\ref{fig:cls}(b). 
All the currents $\mathcal{J}$ are non-zero, except for $\phi=0,\pm\pi$ where time-reversal symmetry is restored in the lattice.

Importantly, the total current inside each of these plaquettes is invariant under basis redefinitions, that preserve TRS.
This is an important result, given the FB degeneracy and the absence of a preferred basis for FB eigenstates.
To show this, we use the Hellmann--Feynman theorem, that connects the derivative of the eigenvalue with respect to a parameter to the expectation value of the derivative of the associated Hamiltonian with respect to the same parameter
\begin{gather}
    \dv{E_\lambda}{\lambda} = \ev**{\dv{H(\lambda)}{E_\lambda}}{E_\lambda}\,.
\end{gather}
\(\ket{E_\lambda}\) is the normalized eigenstate of the parameter-dependent Hamiltonian $H(\lambda)$ corresponding to the eigenvalue $E_\lambda$.
The parameter in our picture is the flux $\lambda=\phi$ which enters the Hamiltonian exclusively through plaquette ($P$) hoppings, $[H_0]_{\mu\nu} = e^{\imi \xi_{\mu\nu} \phi/M}$ for $\langle \mu,\nu \rangle \in P$, and $\xi_{\mu\nu} = +1$ if the bonds are traversed in clockwise, and $-1$ if counter-clockwise.  The entangling unitary $\bmu$ leaves both the plaquette block and the plaquette components of each CLS unchanged, so $\partial_\phi H$ maps the support of $\ket{\psi_n}$ into itself and does not couple distinct CLSs. Hence $\{\ket{\psi_n}\}$ diagonalizes $\partial_\phi H$ within the degenerate flat-band subspace.
So, we have $\pdv*{[H_0]_{\mu\nu}}{\phi} = (\imi \xi_{\mu\nu}/M)[H_0]_{\mu\nu}$ for plaquette bonds, and zero everywhere else.
Therefore, we can write the following expression
\begin{gather}
    \pdv{E_n}{\phi} = \ev**{\pdv{H_0}{\phi}}{\psi_{n}} = \frac{1}{M}\sum_{\langle \mu,\nu \rangle \in P}\mathcal{J}_{\mu\nu} = \frac{1}{M}\mathcal{J}_P\, ,
\end{gather}
where $\mathcal{J}_P$ denotes the total current ${\cal J}$ in plaquette $P$.
Given the invariance of eigenvalues under basis changes, we see that the total current is also basis independent.



\section{Generic flat band lattices with broken time reversal symmetry}
\label{sec:construction_FB}


In this section, we focus on constructing Hamiltonian lattices $H$ in Eq.~\eqref{eq:H} with broken time-reversal symmetry featuring coexisting dispersive and flat bands. 
In other words, only a subset of the Bloch bands is independent of momentum ${\bf k}$ over the entire Brillouin zone, $E_\mu({\bf k}) = \mathrm{const}$ for $1\leq \mu\leq \ell< \nu$. 
In this case however, differently from the ABF lattices discussed in the previous section, the CLS of a flat band can either form an orthogonal complete set, or form a non-orthogonal set~\cite{kim2026real}. 
On one hand, orthogonal CLS sit within single unit cells of the flat band lattices, and they can be detangled from the remaining dispersive states via local unitary rotations $\bmu$~\cite{flach2014detangling}.
On the other hand, non-orthogonal CLS span over two or more unit cells, and they cannot be all detangled from the dispersive states via local unitary rotations~\cite{leykam2017localization}. 
Therefore, in the following, we treat the cases of orthogonal and non-orthogonal CLS separately.

\subsection{Orthogonal CLS}
\label{sec:orthoCLS}

Let us start with flat band lattices supporting orthogonal CLS.
In this case, analogous to Sec.~\ref{sec:construction}, the construction of flat band lattices with broken TRS relies on unitary rotations $\bmu$ that entangle disconnected flat band states into a dispersive lattice~\cite{flach2014detangling}. 
The scheme begins by arranging flux-threaded plaquettes adjacently to a dispersive lattice.
The total number of bands therefore is $M+\nu_d$, where $M$ is the size of the plaquette, while $\nu_d$ is the number of dispersive bands. 
We then apply a unitary transformation $\bmu$ formed of local rotations $U$ that entangle each plaquette with one unit-cell of the dispersive lattice~\cite{flach2014detangling}.
This generally results in a complex lattices, which again can be simplified by attaching auxiliary sites to each plaquette. 

\begin{figure}
    \centering
    \includegraphics[width=\columnwidth]{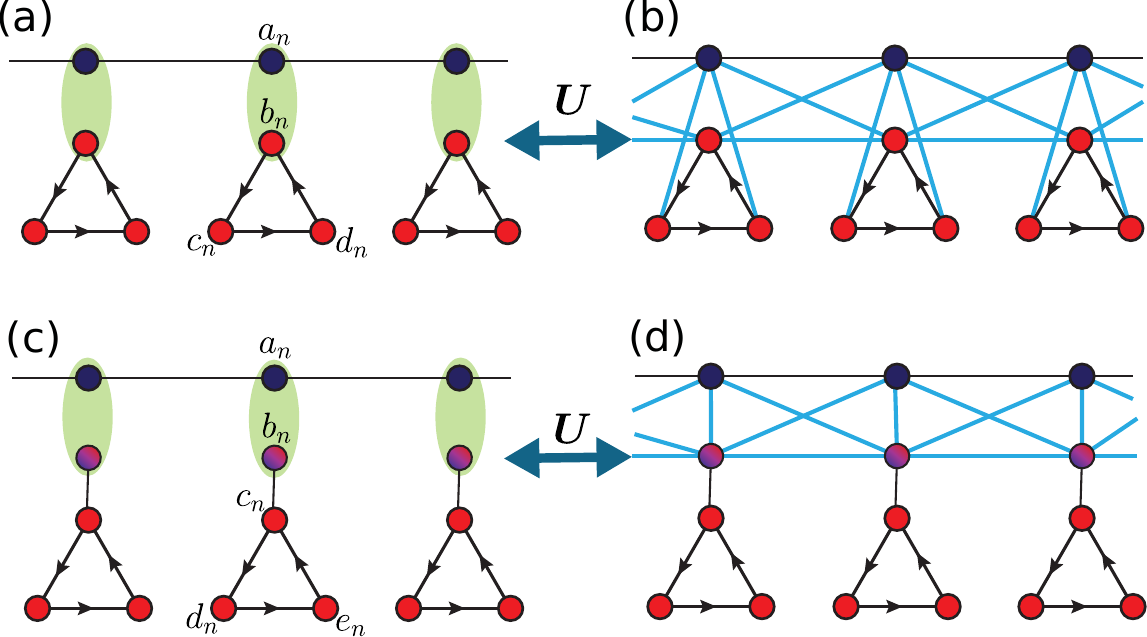}
    \caption{
        (a) 1D lattice formed by disconnected flux-threaded triangular plaquettes adjacent to a one-dimensional dispersive lattice.
        Green shadings mark the sites to which local unitary operation is applied.
        (b) Resulting lattice after the entangling rotations.
        (c),~(d) same as (a),~(b) for triangular plaquettes with auxiliary sites. 
    }
    \label{fig:FB_ortho1}
\end{figure}

We show an example in Fig.~\ref{fig:FB_ortho1}(a), which follows the simplest case in Sec.~\ref{sec:construction_ABF_1d}. 
In this case, we consider flux-threaded triangular plaquettes in Eq.~\eqref{eq:1d_trig}, shaded in blue color, which we arranged adjacently to a one dimensional dispersive chain. 
The intra-cell and inter-cell matrices $H_0$ and $H_1$ of this lattice in the detangled representation are
\begin{gather}
\begin{split}
    H_0 =
    \begin{pmatrix}
        0 & 0 & 0 & 0\\
        0 & 0 & e^{\imi \frac{\phi}{3} } & e^{-\imi \frac{\phi}{3}} \\
        0 & e^{-\imi \frac{\phi}{3}} & 0 & e^{\imi \frac{\phi}{3}} \\
        0 &  e^{\imi \frac{\phi}{3}} & e^{-\imi \frac{\phi}{3}} & 0 
    \end{pmatrix}
    \quad
    H_1 =  
    \begin{pmatrix}
        1 & 0 & 0 & 0\\
        0 & 0 & 0 & 0  \\
        0 & 0 & 0 & 0\\
        0 &  0 & 0 & 0
    \end{pmatrix}
\end{split}
\label{eq:Ham_eq_lin}
\end{gather}
We then design a unitary transformation $\bmu$ formed by local rotations \(U\!\in\!\gsun{2}\) in Eq.~\eqref{eq:SO2_U} defined between pairs of sites, as highlighted with green shared areas.
This yields an entangled lattice with four bands---three of which are flat whose values are in Eq.~\eqref{eq:ABF_tr}, while the fourth one is dispersive at $E(k) = 2\cos k$.
The lattice is shown in Fig.~\ref{fig:FB_ortho1}(b), whose matrices are
\begin{align}
\begin{split}
    & H_0 =
    \begin{pmatrix}
        0 & 0 & e^{\imi \frac{\phi}{3} } \sin\theta & e^{-\imi \frac{\phi}{3} }  \sin\theta\\
        0 & 0 & e^{\imi \frac{\phi}{3} } \cos\theta  & e^{-\imi \frac{\phi}{3}} \\
        e^{-\imi \frac{\phi}{3} } \sin\theta & e^{-\imi \frac{\phi}{3}} \cos\theta & 0 & e^{\imi \frac{\phi}{3}} \\
        e^{\imi \frac{\phi}{3} } \sin\theta &  e^{\imi \frac{\phi}{3}} \cos\theta & e^{-\imi \frac{\phi}{3}} & 0 
    \end{pmatrix}
    \\
    &\qquad H_1 =  
    \begin{pmatrix}
        \cos^2\theta  & - \cos\theta\sin\theta & 0 & 0\\
        -\cos\theta\sin\theta & \sin^2\theta & 0 & 0  \\
        0 & 0 & 0 & 0\\
        0 &  0 & 0 & 0
    \end{pmatrix}
\end{split}
\label{eq:Ham_eq_lina}
\end{align}
%

The resulting intricate hopping profile of the generated lattice can be simplified, alike in Sec.~\ref{sec:construction_ABF_1d}, by extending each plaquette with an auxiliary site, as shown in Fig.~\ref{fig:FB_ortho1}(c).
The matrices of the SD lattice with the augmented plaquettes are 
\begin{gather}
\begin{split}
    H_0 =
    \begin{pmatrix}
        0 & 0 & 0 & 0 & 0\\
        0 & 0 & 1 & 0 & 0\\
        0 & 1 & 0 & e^{\imi \frac{\phi}{3} } & e^{-\imi \frac{\phi}{3}} \\
        0 & 0 & e^{-\imi \frac{\phi}{3}} & 0 & e^{\imi \frac{\phi}{3}} \\
        0 & 0 &  e^{\imi \frac{\phi}{3}} & e^{-\imi \frac{\phi}{3}} & 0 
    \end{pmatrix}
    \quad
    H_1 =  
    \begin{pmatrix}
        1 & 0 & 0 & 0 & 0\\
        0 & 0 & 0 & 0 & 0\\
        0 & 0 & 0 & 0 & 0\\
        0 & 0 & 0 & 0 & 0\\
        0 & 0 & 0 & 0 & 0
    \end{pmatrix}
\end{split}
\label{eq:Ham_eq_linb}
\end{gather}
The entangling process between the dispersive chain and the auxiliary sites, highlighted with green shared areas within each unit-cell, yields the entangled lattice shown in Fig.~\ref{fig:FB_ortho1}(d), whose matrices are
\begin{align}
\begin{split}
    H_0 &=
    \begin{pmatrix}
        0 & 0 & \sin\theta & 0 & 0 \\ 
        0 & 0 & \cos\theta & 0 & 0\\
        \sin\theta & \cos\theta & 0 & e^{\imi \frac{\phi}{3} }  & e^{-\imi \frac{\phi}{3}} \\
        0 & 0 & e^{-\imi \frac{\phi}{3}}  & 0 & e^{\imi \frac{\phi}{3}} \\
        0 & 0 & e^{\imi \frac{\phi}{3}} & e^{-\imi \frac{\phi}{3}} & 0 
    \end{pmatrix}
    \\
    H_1 &=  
    \begin{pmatrix}
        \cos^2\theta  & - \cos\theta\sin\theta & 0 & 0 & 0 \\
        -\cos\theta\sin\theta & \sin^2\theta & 0 & 0 & 0 \\
        0 & 0 & 0 & 0 & 0\\
        0 &  0 & 0 & 0 & 0 \\
        0 &  0 & 0 & 0 & 0 
    \end{pmatrix}
\end{split}
\label{eq:Ham_eq_linc}
\end{align}
In this case, the flux-threaded triangular plaquettes, where the non-zero flux currents are located, remain unchanged. 
Let us observe that this construction can not only be extended by increasing the plaquette size and changing the unitary rotations $\bmu$, but also by considering a different dispersive lattice set adjacently to the plaquettes.

\subsection{Non-orthogonal CLS}
\label{sec:N_orthoCLS}

Flat bands supporting non-orthogonal CLS with broken TRS and non-zero flux are much harder to find, since systematic construction schemes are harder to come by~\cite{katsura2010ferromagnetism} (however, see the recently developed approach using Bloch compact localized states~\cite{kim2026real}).

The ABF case discussed in Sec.~\ref{sec:construction} and the orthogonal FB case discussed in Sec.~\ref{sec:construction_FB} relies on local unitary rotations \(\bmu\) to construct flat bands with fluxes.
Such rotations---that detangle the system into disjoint clusters---are absent in the case of non-orthogonal CLS; the previous ABF construction becomes impossible.

\begin{figure}
    \centering
    \includegraphics[width=\columnwidth]{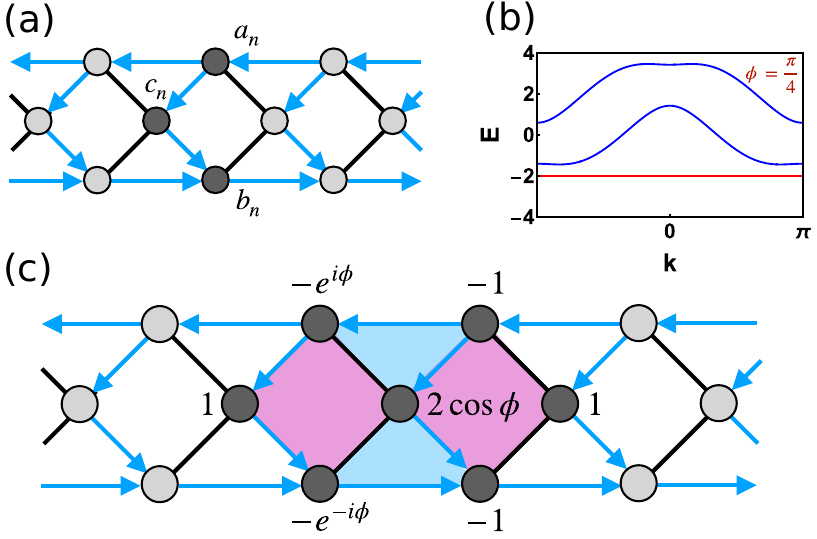}
    \caption{
        (a) One dimensional flat band lattice with four-site flux-threaded plaquette. 
        (b) Band structure for flux $\phi=\frac{\pi}{4}$.
        The $E=-2$ flat band is colored in red, while the two dispersive bands are colored in blue.  
        (c) CLS as function of the flux $\phi$. The pink and blue shaded regions indicate where the currents circulate, with the color code denoting opposite directions.
    }
    \label{fig:FB_non_ortho1}
\end{figure}

Nevertheless, such flat bands do exist and can be constructed, although their systematic construction remains an open issue for future research.

An example is displayed in Fig.~\ref{fig:FB_non_ortho1}(a) and it is defined by intra-cell and inter-cell matrices
\begin{gather}
\begin{split}
    H_0 =
    \begin{pmatrix}
        0 & 0 & e^{\imi \phi} \\
        0 & 0 & e^{-\imi \phi} \\
        e^{-\imi \phi} & e^{\imi \phi} & 0 
    \end{pmatrix}
    ,  \quad
    H_1 = 
    \begin{pmatrix}
        e^{-\imi \phi}  & 0 & 1 \\
        0 & e^{\imi \phi}  & 1 \\
        0 & 0 & 0 
    \end{pmatrix}
\end{split}
\label{eq:Ham_eq_line}
\end{gather} 
This lattice, as detailed in Appendix~\ref{app:N_orthoCLS}, is obtained by breaking the chiral symmetry of a diamond chain with the horizontal hopping within the $a$ and the $c$ sub-lattices.
Then, the lattice's plaquettes are thread with carefully oriented magnetic fields, whose flux is controlled by $\phi$. 
This procedure renders one flat band at $E=- 2 $ and two dispersive bands $E_{2,3} =  1+2 \cos\phi \cos k  \pm  (1- \cos (2 k) - 2 \cos (2\phi) \sin^2 k)^{\frac{1}{2}} $ --- shown in Fig.~\ref{fig:FB_non_ortho1}(b) for $\phi=\frac{\pi}{4}$.  
The $\phi$-dependent CLS associated to the flat band is shown in Fig.~\ref{fig:FB_non_ortho1}(c).
The flux induces currents circulating within the plaquettes shaded in blue and pink colors, which respectively $\mathcal{J} = -6\sin 2\phi$ and  $\mathcal{J} = 4\sin 2\phi$.
The total current in the CLS is given by double the sum of these two currents and is non-zero.

We note that a flat band with a similar CLS was proposed earlier in Ref.~\cite{katsura2010ferromagnetism}.

Finally, we note that the above example illustrates the existence of flat bands with non-orthogonal CLS with fluxes, but does not hint at a systematic construction of such flat bands.
In the absence of local unitary transformations that detangle the Hamiltonian a different strategy is required.
We leave its development to future work.


\subsection{Chiral flat band lattices with broken TRS}


An important class of flat-band systems is formed by chiral FBs, where the existence of the flat band is guaranteed by the chiral (sublattice) symmetry of the lattice.
This class has been explored extensively in Refs.~\cite{vidal1998aharonov,vidal2000interaction,vidal2001disorder,khomeriki2016landau,ramachandran2017chiral}.
In the paradigmatic example of Vidal \emph{et al.}~\cite{vidal1998aharonov}, complex hoppings on the diamond lattice generate Hamiltonians in which TRS is broken.
Furthermore, the CLS at the flat-band energy is non-orthogonal for a general flux threaded through the lattice~\citet{khomeriki2016landau}.

In such systems, the CLS has non-zero amplitude only on the majority sublattices $(a_n,b_n)$ and fails to occupy all sites of any flux-threaded plaquette.
As a result, no gauge-invariant current loop can be formed, and each plaquette contributes zero net current.
In the remaining cases where the CLS does occupy all sites of a flux-threaded polygon, either
(i) its amplitudes are entirely real---yielding no current; or 
(ii) the corresponding flux configuration preserves TRS, placing it outside the scope of the present work.

Thus, chiral flat bands with broken TRS differ fundamentally from both the ABF constructions and the more general flat-band examples developed here: they cannot support compact current vortices. 

\section{Discussion and Conclusion}
\label{sec:discussion}

We have provided a protocol of generating all-bands-flat lattices with broken time reversal symmetry.
We introduced the TRS breaking by threading magnetic flux through polygonal plaquettes ($M$ sites).
Although we focused our attention on triangle and square plaquettes, our protocol can be trivially extended to any polygonal plaquettes.
Besides, we show that just by selecting the plaquettes and using entangling unitaries, we can
(a) preserve the broken time reversal symmetry, and
(b) transform the semi-detangled lattice to non-detangled lattice;
we end up with ABF lattices with non-trivial hopping and complicated network.
Therefore, we introduce appendices to each of the vertices of the the plaquettes which allow us to braid an ABF lattice in 1D (Fig.~\ref{fig:1D_ABF}), 2D (Figs.~\ref{fig:2D_ABF}, \ref{fig:ABF}), and 3D (Fig.~\ref{fig:3D_ABF}) as well as leave the flux threaded plaquettes unspoiled.
However, the introduction of appendices increase the number of bands twofold, \emph{i.e.,} $2M$, with the benefit of easier extendability to higher dimensions.
We can stack $L = 2,3, ... $ augmented plaquettes allowing us to reduce the complex hopping junctions to rhombic hopping. 
This results in families of physically relevant ABF lattices with broken TRS with $2 M L$ bands.
The CLS in these systems occupy the initial plaquette and auxiliary sites, along with the entangled sites from neighboring unit cells involved in the unitary transformation $\bmu$.
The CLS components residing in the flux-threaded plaquettes generate nonzero local currents, effectively behaving as localized lattice vortices.
These localized vortices are also basis-independent, as long as the underlying transformation does not change the TRS status of the Hamiltonian.


We further show that our construction is also implementable to systems with flat-band in presence of dispersive bands:
We consider the cases of orthogonal (Fig.~\ref{fig:FB_ortho1}), and non-orthogonal (Fig.~\ref{fig:FB_non_ortho1}) CLSs.
We find our broken TRS systems contain localized circulation, however, unlike ABF or orthogonal cases, the non-orthogonal case has current in all the plaquettes as seen from the diamond lattice(Fig.~\ref{fig:FB_non_ortho1}(c)).

ABF lattices are by construction fine-tuned: the hopping parameters must satisfy precise relations to enforce destructive interference across all bands.
This is a generic feature of any flat-band system with short-range hopping and is not specific to the TRS-broken constructions presented here.
We expect that fabrication imperfections or disorder, inevitably present in experimental realizations, would make bands dispersive or finite width, removing the CLS.
For weak perturbations we might still expect well-localised currents at the locations of CLS.
The full analysis of this problem is an interesting open issue, that is beyond the scope of the present work.

The 1D and 2D models presented here could be implemented in several experimental settings. 
In electrical circuit networks the magnetic flux threading plaquettes can be emulated via a capacitive-inductive coupling that introduces a phase, as has been demonstrated for Aharonov-Bohm cage lattices in Refs.~\cite{chase-mayoral_2023_compact,lape2025realization}, also discussed in the context of flat-band circuit networks~\cite{zhou_2025_anyonic, zhang_2023_nonabelian}.
In circuits, the auxiliary (appendix) sites are just additional nodes in the network, making the geometrical extension to our augmented unit cells straightforward.
The main engineering challenge is to achieve the required phase accuracy in the hopping parameters; state-of-the-art circuit fabrication has demonstrated sub-percent precision, which is sufficient for observing caging and CLS-localized currents.
Photonic waveguide arrays are a common setup to implement flatbands~\cite{vicencio2015observation,mukherjee2015observation,weimann2016transport}.
Complex hoppings are realized via evanescent coupling with engineered mode overlap profiles~\cite{vicencio2015observation, weimann2016transport}.
Synthetic gauge fields~\cite{mukherjee2018experimental} have been realized using helical waveguide modulation~\cite{rechtsman_2013_PhotonicFloquet,yang_2020_PhotonicFloquet}.
The probability current circulating in a CLS plaquette can be directly mapped onto the optical intensity flow in the waveguide cross-sections~\cite{kremer2020square,Jorg2020artificial,caceres2022controlled}.
Cold atom setups with synthetic gauge fields offer exact flux control~\cite{kang2020creutz,li2022aharonov}, though implementing the auxiliary site structure and multi-band architecture would require complex optical lattice engineering.
The 3D models are more challenging experimentally.
In photonic waveguide arrays, the propagation coordinate plays the role of time, restricting static chips to two effective lattice dimensions; realizing our 3D constructions in photonics would therefore require synthetic dimensions~\cite{lustig_2019_Photonictopological}.
Electrical circuits, in contrast, are not constrained by the embedding geometry:
the lattice dimensionality is set 
purely by connectivity, and even four-dimensional lattices have been implemented~\cite{wang_2020_Circuitimplementation}.
This makes circuit networks the most promising platform for the tetrahedral constructions of Fig.~\ref{fig:3D_ABF}.

Our work opens up the possibility of studying localized currents in otherwise dispersionless lattice.
Since these currents are dependent on the strength of the flux threaded through the plaquettes, it allows tunability, which can be of interest to the experimentalists of the subject. 

\section*{Acknowledgments}
    RKR, AA, and SF acknowledge the support by the Institute for Basic Science in Korea (IBS-R024-D1).
    CD acknowledges the support  the PNRR MUR project CN 00000013-ICSC. 
    We thank Tilen \v{C}ade\v{z}, Dung Xuan Nguyen and Yeongjun Kim for valuable discussions.

\appendix

\section{TRS versus unitary transformations}
\label{appA:TRS}

We consider a unitary operator $\bmu \in \Hil$ such that
\begin{gather}
    \label{eqA:hamilton_transform}
    H^\prime := \bmu H \bmu^\dagger.
\end{gather}
Using the definition in Eq.~\eqref{eq:trs_def} in the main text, we want to explore the conditions which will lead to $H^\prime$ be TRS, \emph{i.e.,}
\begin{gather}
    \label{eqA:hprime_trs}
    \mathcal{T} H^\prime \mathcal{T}^{-1} = H^\prime.
\end{gather}

Let us revisit the properties of anti-unitary transformations.
Consider the following lemma ---

\begin{lemma}
Let $\mathcal{T}$ be an anti-unitary operator and $A$ a bounded linear operator on $\Hil$.
Then for any $\ket{\phi}, \ket{\psi} \in \Hil$,
\begin{gather}
    \label{eqA:antiunitary-action}
    \qexp{\mathcal{T} \psi}{\mathcal{T} A \mathcal{T}^{-1}}{\mathcal{T} \phi} = \left(\qexp{\phi}{A}{ \psi}\right)^\ast.
\end{gather}
\end{lemma}
\begin{proof}
    Let $A^\prime = \mathcal{T}A\mathcal{T}^{{-}1}$, then, 
    \begin{gather}
        \qexp{\mathcal{T}\psi}{A^\prime}{\mathcal{T}\phi} = \qexp{\mathcal{T}\psi}{\mathcal{T}A}{\phi} = \left(\qexp{\phi}{A}{\psi}\right)^\ast,
    \end{gather}
    where in the last step we have used the anti-linearity and anti-unitarity properties of $\mathcal{T}$.
\end{proof}
From this lemma, we note the following:
\begin{gather}
    \left(\mathcal{T}A\mathcal{T}^{{-}1}\right)^\dagger = \mathcal{T}A^\dagger\mathcal{T}^{{-}1}.
\end{gather}

We now state the main result of this appendix.

\begin{theorem}
\label{prop:TRS}
Let $H$ be TRS, i.e., Eq.~\eqref{eq:trs_def} holds, and define $H^\prime$ via Eq.~\eqref{eqA:hamilton_transform}.
Then $H^\prime$ is TRS if and only if
\begin{gather}
    \label{eqA:trs-preserve}
    \comm{H}{\mathcal{W}} = 0, \quad \mathcal{W}:= \bmu^\dagger(\mathcal{T}\bmu\mathcal{T}^{-1}).
\end{gather}
\end{theorem}

\begin{proof}
Considering $H$ with TRS, we write:

\begin{align}
\begin{split}
    \mathcal{T} H^\prime \mathcal{T}^{-1}
    &= \mathcal{T} (\bmu H \bmu^\dagger) \mathcal{T}^{-1} \\
    &= (\mathcal{T} \bmu \mathcal{T}^{-1})(\mathcal{T} H \mathcal{T}^{-1})(\mathcal{T} \bmu^\dagger \mathcal{T}^{-1}) \\
    &= (\mathcal{T} \bmu \mathcal{T}^{-1}) H (\mathcal{T} \bmu \mathcal{T}^{-1})^\dagger.
\end{split}
\end{align}

Thus, $\mathcal{T} H^\prime \mathcal{T}^{-1} = H^\prime$ if and only if (using Eq.~\eqref{eqA:hamilton_transform} on the rhs)
\begin{gather}
    \label{eqA:trs_first}
    (\mathcal{T} \bmu \mathcal{T}^{-1}) H (\mathcal{T} \bmu \mathcal{T}^{-1})^\dagger = \bmu H \bmu^\dagger.
\end{gather}

Now, we define $\mathcal{W}:= \bmu^\dagger(\mathcal{T}\bmu\mathcal{T}^{-1})$. Using this we can write Eq.~\eqref{eqA:trs_first} as 
\begin{gather}
\begin{split}
    \mathcal{W}H\mathcal{W}^\dagger &= H,\\
    \mathcal{W}H &= H\mathcal{W},
\end{split}
\end{gather}
which implies
\begin{gather}
    \comm{H}{\mathcal{W}} = 0.
\end{gather}
\end{proof}

\begin{corollary}
    It follows from Theorem~\ref{prop:TRS}, a given unitary transformation $\bmu$ preserves the TRS of Hamiltonian $H$ if it is itself TRS preserving, \emph{i.e.,}
    \begin{gather}
    \label{eqA:trs_simple}
        \mathcal{T}\bmu\mathcal{T}^{-1} = \bmu.
    \end{gather}
\end{corollary}
\begin{proof}
    Follows from the main result of Theorem~\ref{prop:TRS}, Eq.~\eqref{eqA:trs-preserve}. The commutation is trivially satisfied if $\mathcal{W}= \mathrm{I}$ which immediately implies 
    $\mathcal{T}\bmu\mathcal{T}^{-1} = \bmu$.
\end{proof}



\subsection{Spinless Case and orthogonal transformations}

For spinless systems, the time-reversal operator is simply complex conjugation: $\mathcal{T} = K$.
In this case,
\[
    \mathcal{T} \bmu \mathcal{T}^{-1} = \mathcal{K} \bmu \mathcal{K} = \bmu^\ast,
\]
so the condition~\eqref{eqA:trs_simple} becomes:
\begin{gather}
    \bmu = \bmu^\ast, 
    \label{eq:real-unitary}
\end{gather}
implying $\bmu$ to be a real unitary matrix, i.e. orthogonal matrix.
Therefore, in spinless systems, TRS is preserved under orthogonal transformations.
The overall conclusion is that time-reversal symmetry (TRS) is unaffected by orthogonal transformations.


\section{Construction of the unitaries}
\label{appA:unitary}

We provide a general protocol for the construction of unitaries that introduce local rotations between the sites of different unit cells.
We consider the local operations between two-sites and three-sites which can be extended to higher number of such sites.

\subsection{Between two sites}
\label{appA:unitary_su2}

Since the local operations only introduce rotation between two sites (of neighboring unit-cells), we consider an element of \(\gsun{2}\):
\begin{gather}
    U = 
    \begin{pmatrix}
        z & w \\
        -w^\ast & z^\ast
    \end{pmatrix} 
\end{gather} 
parametrize by two complex numbers $z$ and $w$ such that $|z|^2 + |w|^2 =1$. 
We construct \( \bm{U}_v \) along the \( v \)-direction following a systematic protocol. 
Consider a lattice with \( N_x \) unit cells in the \( x \)-direction and \( N_y \) unit cells in the \( y \)-direction, where each unit cell contains \( \nu \) sites.
We initiate \( \bm{U}_v \) by setting it equal to the identity matrix of size \( (\nu N_x N_y) \times (\nu N_x N_y) \).

\begin{enumerate}
    \item For each pair of sites \( ( i, j ) \) located in adjacent unit cells \(\langle\bm{\mu},\bm{\mu}'\rangle\) along the \( v \)-direction:
    \begin{enumerate}
        \item Diagonal Elements:
        \begin{gather}
            \label{eqA:diagonal}
            [\bm{U}_v]_{\bm{\mu},i;\bm{\mu},i} = [U]_{11}, \quad [\bm{U}_v]_{\bm{\mu}',j;\bm{\mu}',j} = [U]_{22}.
        \end{gather}
        \item Off-Diagonal Elements:
        \begin{gather}
            \label{eqA:off_diagonal}
            [\bm{U}_v]_{\bm{\mu}',j;\bm{\mu},i} = [U]_{21}, \quad [\bm{U}_v]_{\bm{\mu},i;\bm{\mu}',j} = [U]_{12}.
        \end{gather}
    \end{enumerate}
    \item For periodic boundary conditions, we apply the same substitutions to unit-cell pairs at both ends along the \( v \)-direction to ensure continuity:
    \begin{enumerate}
        \item Diagonal Elements:
        \begin{gather}
            \label{eqA:diagonal_pbc}
            [\bm{U}_v]_{\bm{\mu},i;\bm{\mu},i} = [U]_{11}, \quad [\bm{U}_v]_{\bm{\mu}',j;\bm{\mu}',j} = [U]_{22}.
        \end{gather}
        \item Off-Diagonal Elements:
        \begin{gather}
            \label{eqA:off_diagonal_pbc}
            [\bm{U}_v]_{\bm{\mu}',j;\bm{\mu},i} = [U]_{21}, \quad [\bm{U}_v]_{\bm{\mu},i;\bm{\mu}',j} = [U]_{12}.
        \end{gather}
    \end{enumerate}
\end{enumerate}

The process is repeated for each direction \( v \) as follows:
\begin{enumerate}
    \item For \( v = x \)-direction (\(\bm{\mu} = m,n;~ \bm{\mu}'= m',n\)):
    \begin{enumerate}
        \item Step 1 (Eqs. \eqref{eqA:diagonal} – \eqref{eqA:off_diagonal}) is repeated \( N_x - 1 \) times for each fixed \( n = 1, \dots, N_y \).
        \item Step 2 (Eqs. \eqref{eqA:diagonal_pbc} – \eqref{eqA:off_diagonal_pbc}) is applied between unit cells at \( m=1 \) and \( m^\prime=N_x \), ensuring periodicity (\(\bm{\mu}' = N_x,n;~ \bm{\mu}= 1,n\)).
    \end{enumerate}
    \item For \( v = y \)-direction (\(\bm{\mu} = m,n;~ \bm{\mu}'= m,n'\)):
    \begin{enumerate}
        \item Step 1 (Eqs. \eqref{eqA:diagonal} – \eqref{eqA:off_diagonal}) is repeated \( N_y - 1 \) times for each fixed \( m = 1, \dots, N_x \).
        \item Step 2 (Eqs. \eqref{eqA:diagonal_pbc} – \eqref{eqA:off_diagonal_pbc}) is applied between unit cells at \( n=1 \) and \( n^\prime=N_y \), ensuring periodicity (\(\bm{\mu}' = m,N_y;~ \bm{\mu}= m,1\)).
    \end{enumerate}
\end{enumerate}

\subsection{Between three sites}
\label{appA:unitary_su3}

The local operations will be formed by an element of \(\gsun{3}\).
However, there are 8 real parameters (\( \alpha_i\) for \( i = 1 \cdots 8\)) that can explicitly describe \(\gsun{3}\), it can be written as follows.
Using \( 3\times 3\) Gell-Mann matrices \cite{gell-mann_2000_eightfold} \( \lambda_i \) for \( i = 1\cdots 8\), the general unitary can be written as:
\begin{gather}
    U = \exp(\imi \sum_{i=1}^8 \alpha_i \lambda_i). 
\end{gather}
However, for the demonstration of our construction, we use a real \(\mathrm{SO}(3)\) (subgroup of \(\gsun{3}\)) matrix, parametrized by three Euler angles \((\theta, \phi, \psi)\) as:
\begin{equation}
    U = \begin{pmatrix}
        cp\,ct\,cps-sp\,sps & -cp\,ct\,sps-sp\,cps & cp\,st\\
        sp\,ct\,cps+cp\,sps & -sp\,ct\,sps+cp\,cps & sp\,st\\
        -st\,cps & st\,sps & ct
    \end{pmatrix},
\end{equation}
where we have used the shorthands of the form \(\#1\#2\) where in \(\#1\) we place \( c\) for \(\cos\), \( s\) for \( \sin\); and in \(\#2\) \( p, t, ps\) for \(\phi, \theta,\) and \(\psi\), respectively.
Now that we have given an example of the kind of unitaries we need, let us look at how to construct \(\bm{U}_v\). For sites \(i, j,\) and \(k\) of unit-cells \( \langle \bm{\mu}, \bm{\mu}^\prime, \bm{\mu}^{\dprime}\rangle\), we do the following, 

\begin{enumerate}
    \item Diagonal Elements:
    \begin{align}
    \label{eqA:diagonal_su3}
    \begin{split}
        [\bm{U}_v]_{\bm{\mu},i;\bm{\mu},i} = [U]_{11}, \quad& [\bm{U}_v]_{\bm{\mu}',j;\bm{\mu}',j} = [U]_{22}, \\ [\bm{U}_v]_{\bm{\mu}^{\dprime},k;\bm{\mu}^{\dprime},k} = [U]_{33} &.
    \end{split}
    \end{align}
    \item Off-Diagonal Elements:
    \begin{align}
    \label{eqA:off_diagonal_su3}
    \begin{split}
        [\bm{U}_v]_{\bm{\mu},i;\bm{\mu}^\prime,j} = [U]_{12}, \quad& [\bm{U}_v]_{\bm{\mu},i;\bm{\mu}^{\dprime},k} = [U]_{13},\\
        [\bm{U}_v]_{\bm{\mu}^\prime,j;\bm{\mu}^{\dprime},k} = [U]_{23}, \quad& [\bm{U}_v]_{\bm{\mu}^\prime,j;\bm{\mu},i} = [U]_{21},\\
        [\bm{U}_v]_{\bm{\mu}^{\dprime},k;\bm{\mu},i} = [U]_{31}, \quad& [\bm{U}_v]_{\bm{\mu}^{\dprime},k;\bm{\mu}^{\prime},j} = [U]_{32}.
    \end{split}
    \end{align}
\end{enumerate}
For periodic boundary conditions, we apply the same substitutions to unit-cell pairs at both ends along the \( v \)-direction to ensure continuity:
\begin{enumerate}
    \item Diagonal Elements:
    \begin{align}
    \label{eqA:diagonal_pbc_su3}
    \begin{split}
        [\bm{U}_v]_{\bm{\mu},i;\bm{\mu},i} = [U]_{11}, \quad& [\bm{U}_v]_{\bm{\mu}',j;\bm{\mu}',j} = [U]_{22}, \\ [\bm{U}_v]_{\bm{\mu}^{\dprime},k;\bm{\mu}^{\dprime},k} = [U]_{33} &.
    \end{split}
    \end{align}
    \item Off-Diagonal Elements:
    \begin{align}
    \label{eqA:off_diagonal_pbc_su3}
    \begin{split}
        [\bm{U}_v]_{\bm{\mu},i;\bm{\mu}^\prime,j} = [U]_{12}, \quad& [\bm{U}_v]_{\bm{\mu},i;\bm{\mu}^{\dprime},k} = [U]_{13},\\
        [\bm{U}_v]_{\bm{\mu}^\prime,j;\bm{\mu}^{\dprime},k} = [U]_{23}, \quad& [\bm{U}_v]_{\bm{\mu}^\prime,j;\bm{\mu},i} = [U]_{21},\\
        [\bm{U}_v]_{\bm{\mu}^{\dprime},k;\bm{\mu},i} = [U]_{31}, \quad& [\bm{U}_v]_{\bm{\mu}^{\dprime},k;\bm{\mu}^{\prime},j} = [U]_{32}.
    \end{split}
    \end{align}
\end{enumerate}

\section{Construction of non-orthogonal flat bands with broken time reversal symmetry}
\label{app:N_orthoCLS}

In this section, we detail the construction of the flat band lattice with non-orthogonal CLS and broken time-reversal. 
\begin{figure}
    \centering
    \includegraphics[width=\columnwidth]{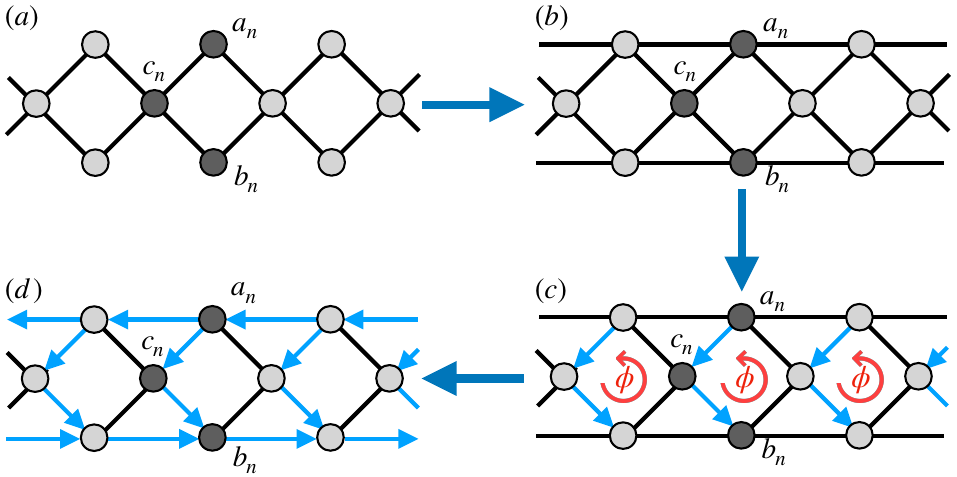}
    \caption{
        (a) Three bands diamond chain.
        The unit-cell sites are colored in black.  
        (b) Diamond chain with a chirality-breaking hopping connecting the nearest-neighboring $a$-sites.  
        (c) Diamond chain with flux threading the central rhombic plaquettes. 
        (d) Diamond chain with flux threading the central rhombic plaquettes and the upper and lower triangular plaquettes. 
    }
    \label{fig:FB_non_ortho1_app}
\end{figure}

The construction starts from a one-dimensional diamond chain shown in Fig.~\ref{fig:FB_non_ortho1_app}(a), whose intra-cell and intercell matrices are
\begin{gather}
    \begin{split}
         H_0 = 
        \begin{pmatrix}
            0 & 0 & 1 \\
            0 & 0 & 1 \\
            1 & 1 & 0 \\
        \end{pmatrix}
        \qquad 
        H_1 = 
        \begin{pmatrix}
            0 & 0 & 1 \\
            0 & 0 & 1 \\
            0 & 0 & 0 \\
        \end{pmatrix}
    \end{split}
    \label{eq:diamo1}
\end{gather}
This lattice is chiral~\cite{ramachandran2017chiral}, as it splits in a majority sublattice $\{a_n,b_n\}$ and a minority sublattice $\{c_n\}$ with different cardinality.
This implies that one Bloch band is flat at $E=0$, while the two dispersive bands are $E_{1,2} = \pm 2\sqrt{1 + \cos k}$. 

At first, we break the chirality of the lattice by adding horizontal hopping terms of the same strength of the rhombic hoppings, as shown in Fig.~\ref{fig:FB_non_ortho1_app}(b). 
This turns the three bands of the lattice to $E_1= -2$, $E_2 = 2\cos k$ and $E_{3} =  2( 1+\cos k )$. 

We then thread the central rhombic plaquettes with a magnetic flux, as shown in Fig.~\ref{fig:FB_non_ortho1_app}(c). 
The intra-cell and intercell matrices in Eq.~\eqref{eq:diamo1} turn to 
\begin{gather}
    \begin{split}
        H_0 = 
        \begin{pmatrix}
            0 & 0 & e^{\imi \phi} \\
            0 & 0 & e^{-\imi \phi} \\
            e^{-\imi \phi} & e^{\imi \phi} & 0 
        \end{pmatrix}
        \qquad 
        H_1 = 
        \begin{pmatrix}
            1 & 0 & 1 \\
            0 & 1 & 1 \\
            0 & 0 & 0 \\
        \end{pmatrix}
    \end{split}
    \label{eq:diamo2}
\end{gather}
This however initially destroys the $E=-2$ flat bands, as the three bands are 
$E_1=2 \cos k$ and $E_{2,3} =  \cos k  \pm  \sqrt{4+ 4 \cos k \cos \phi + \cos^2 (k)} $. 

At last, we turn the horizontal chirality-breaking hoppings complex, as shown in Fig.~\ref{fig:FB_non_ortho1_app}(d). 
The intra-cell and intercell matrices in Eq.~\eqref{eq:diamo2} turn to 
\begin{equation}
    \begin{split}
        H_0 =
        \begin{pmatrix}
            0 & 0 & e^{\imi \phi} \\
            0 & 0 & e^{-\imi \phi} \\
            e^{-\imi \phi} & e^{\imi \phi} & 0 
        \end{pmatrix}
        \qquad 
        H_1 = 
        \begin{pmatrix}
            e^{-\imi \phi}  & 0 & 1 \\
            0 & e^{\imi \phi}  & 1 \\
            0 & 0 & 0 \\
        \end{pmatrix}  
    \end{split}
    \label{eq:diamo3}
\end{equation}
The three bands turn to $E_1=- 2 $ and $E_{2,3} =  1+2 \cos\phi \cos k  \pm  \sqrt{1- \cos (2 k) - 2 \cos (2\phi) \sin^2 k} $.

\bibliography{flatband,reference}

@article{kim2026real,
  title = {Real-space decay of flat-band projectors from compact localized states},
  author = {Kim, Yeongjun and Flach, Sergej and Andreanov, Alexei},
  journal = {Phys. Rev. B},
  volume = {113},
  issue = {23},
  pages = {235110},
  numpages = {12},
  year = {2026},
  month = {Jun},
  publisher = {American Physical Society},
  doi = {10.1103/x5rg-lfkb},
  url = {https://link.aps.org/doi/10.1103/x5rg-lfkb},
  eprint={2510.17258},
  archivePrefix={arXiv},
  primaryClass={cond-mat.mes-hall}
}

@article{katsura2010ferromagnetism,
  doi = {10.1209/0295-5075/91/57007},
  url = {https://doi.org/10.1209/0295-5075/91/57007},
  year = {2010},
  month = {sep},
  publisher = {},
  volume = {91},
  number = {5},
  pages = {57007},
  author = {Katsura, Hosho and Maruyama, Isao and Tanaka, Akinori and Tasaki, Hal},
  title = {Ferromagnetism in the Hubbard model with topological/non-topological flat bands},
  journal = {Europhysics Letters},
  eprint = {0907.4564},
  archivePrefix = {arXiv},
  primaryClass = {cond-mat.str-el}
}

@article{yang2023realization,
    author = {Yang, Jing and Li, Yuanzhen and Yang, Yumeng and Xie, Xinrong and Zhang, Zijian and Yuan, Jiale and Cai, Han and Wang, Da-Wei and Gao, Fei},
    year = {2024},
    date = {2024/02/19},
    title = {Realization of all-band-flat photonic lattices},
    journal = {Nature Communications},
    pages = {1484},
    volume = {15},
    issue = {1},
    issn = {2041-1723},
    url = {https://doi.org/10.1038/s41467-024-45580-w},
    doi = {10.1038/s41467-024-45580-w}
}

@article{xia2025fully,
  title = {Fully Flat Bands in a Photonic Dipolar Kagome Lattice},
  author = {Xia, Han-Rong and Wang, Ziyao and Wang, Yunrui and Gao, Zhen and Xiao, Meng},
  journal = {Phys. Rev. Lett.},
  volume = {135},
  issue = {17},
  pages = {176902},
  numpages = {6},
  year = {2025},
  month = {Oct},
  publisher = {American Physical Society},
  doi = {10.1103/bt9s-qsfj},
  url = {https://link.aps.org/doi/10.1103/bt9s-qsfj},
  eprint = {2509.12843},
  archivePrefix = {arXiv},
  primaryClass = {physics.optics},
}

@article{samak2024direct,
  title = {Direct Observation of All-Flat Bands Phononic Metamaterials},
  author = {Samak, Mahmoud M. and Bilal, Osama R.},
  journal = {Phys. Rev. Lett.},
  volume = {133},
  issue = {26},
  pages = {266101},
  numpages = {6},
  year = {2024},
  month = {Dec},
  publisher = {American Physical Society},
  doi = {10.1103/PhysRevLett.133.266101},
  url = {https://link.aps.org/doi/10.1103/PhysRevLett.133.266101}
}

@article{ma2021acoustic,
    author = {Ma, Tian-Xue and Fan, Quan-Shui and Zhang, Chuanzeng and Wang, Yue-Sheng},
    title = {Acoustic flatbands in phononic crystal defect lattices},
    journal = {Journal of Applied Physics},
    volume = {129},
    number = {14},
    pages = {145104},
    year = {2021},
    month = {04},
    issn = {0021-8979},
    doi = {10.1063/5.0040804},
    url = {https://doi.org/10.1063/5.0040804}
}

@article{shen2022observing,
  title = {Observing localization and delocalization of the flat-band states in an acoustic cubic lattice},
  author = {Shen, Ya-Xi and Peng, Yu-Gui and Cao, Pei-Chao and Li, Jensen and Zhu, Xue-Feng},
  journal = {Phys. Rev. B},
  volume = {105},
  issue = {10},
  pages = {104102},
  numpages = {12},
  year = {2022},
  month = {Mar},
  publisher = {American Physical Society},
  doi = {10.1103/PhysRevB.105.104102},
  url = {https://link.aps.org/doi/10.1103/PhysRevB.105.104102}
}

@article{lape2025realization,
  title = {Realization and characterization of an all-bands-flat electrical lattice},
  author = {Lape, Noah and Diubenkov, Simon and English, L. Q. and Kevrekidis, P. G. and Andreanov, Alexei and Kim, Yeongjun and Flach, Sergej},
  journal = {Phys. Rev. B},
  volume = {112},
  issue = {18},
  pages = {184309},
  numpages = {7},
  year = {2025},
  month = {Nov},
  publisher = {American Physical Society},
  doi = {10.1103/1w5c-nsmh},
  url = {https://link.aps.org/doi/10.1103/1w5c-nsmh},
  eprint = {2508.13571},
  archivePrefix = {arXiv},
  primaryClass = {cond-mat.mes-hall}
}

@article{ara2025flat,
	title = {{Flat bands and compact localised states: A Carrollian roadmap}},
	author = {Nisa Ara and Aritra Banerjee and Rudranil Basu and Bhagya Krishnan},
	journal = {SciPost Phys.},
	volume = {19},
	pages = {046},
	year = {2025},
	publisher = {SciPost},
	doi = {10.21468/SciPostPhys.19.2.046},
	url = {https://scipost.org/10.21468/SciPostPhys.19.2.046},
	eprint = {2412.18965},
  archivePrefix = {arXiv},
  primaryClass = {hep-th}
}

@article{martinez2023interaction,
  author = {Jeronimo G. C. Martinez  and Christie S. Chiu  and Basil M. Smitham  and Andrew A. Houck },
  title = {Flat-band localization and interaction-induced delocalization of photons},
  journal = {Science Advances},
  volume = {9},
  number = {50},
  pages = {eadj7195},
  year = {2023},
  doi = {10.1126/sciadv.adj7195},
  url = {https://www.science.org/doi/abs/10.1126/sciadv.adj7195},
  eprint = {2303.02170},
  archivePrefix = {arXiv},
  primaryClass = {quant-ph}
}

@article{kim2023general,
  author = {Kim, Hyeongseop and Oh, Chang-geun and Rhim, Jun-Won},
  year = {2023},
  date = {2023/10/18},
  title = {General construction scheme for geometrically nontrivial flat band models},
  journal = {Communications Physics},
  pages = {305},
  volume = {6},
  issue = {1},
  issn = {2399-3650},
  url = {https://doi.org/10.1038/s42005-023-01407-6},
  doi = {10.1038/s42005-023-01407-6},
  eprint = {2305.00448},
  archivePrefix = {arXiv},
  primaryClass = {cond-mat.str-el}
}

@article{sathe2025topological,
  title = {Topological triviality of flat Hamiltonians},
  author = {Sathe, Pratik and Roy, Rahul},
  journal = {Phys. Rev. B},
  volume = {111},
  issue = {4},
  pages = {L041105},
  numpages = {6},
  year = {2025},
  month = {Jan},
  publisher = {American Physical Society},
  doi = {10.1103/PhysRevB.111.L041105},
  url = {https://link.aps.org/doi/10.1103/PhysRevB.111.L041105},
  eprint = {2309.06487},
  archivePrefix = {arXiv},
  primaryClass = {cond-mat.mes-hall}
}

@article{ryu2024orthogonal,
  doi = {10.1088/1751-8121/ad909d},
  url = {https://dx.doi.org/10.1088/1751-8121/ad909d},
  year = {2024},
  month = {nov},
  publisher = {IOP Publishing},
  volume = {57},
  number = {49},
  pages = {495301},
  author = {Jung-Wan Ryu and Alexei Andreanov and Hee Chul Park and Jae-Ho Han},
  title = {Orthogonal flatbands in Hamiltonians with local symmetry},
  journal = {Journal of Physics A: Mathematical and Theoretical},
  eprint = {2308.14997},
  archivePrefix = {arXiv},
  primaryClass = {cond-mat.mes-hall}
}

@article{danieli2024flat,
  url = {https://doi.org/10.1515/nanoph-2024-0135},
  title = {Flat band fine-tuning and its photonic applications},
  author = {Carlo Danieli and Alexei Andreanov and Daniel Leykam and Sergej Flach},
  pages = {3925--3944},
  volume = {13},
  number = {21},
  journal = {Nanophotonics},
  doi = {doi:10.1515/nanoph-2024-0135},
  year = {2024},
  eprint = {2403.17578},
  archivePrefix = {arXiv},
  primaryClass = {physics.optics},
}

@article{regnault2022topological,
  author = {Regnault, Nicolas and Xu, Yuanfeng and Li, Ming-Rui and Ma, Da-Shuai and Jovanovic, Milena and Yazdani, Ali and Parkin, Stuart S. P. and Felser, Claudia and Schoop, Leslie M. and Ong, N. Phuan and Cava, Robert J. and Elcoro, Luis and Song, Zhi-Da and Bernevig, B. Andrei},
  year = {2022},
  date = {2022/03/01},
  title = {Catalogue of flat-band stoichiometric materials},
  journal = {Nature},
  pages = {824--828},
  volume = {603},
  issue = {7903},
  issn = {1476-4687},
  url = {https://doi.org/10.1038/s41586-022-04519-1},
  doi = {10.1038/s41586-022-04519-1},
  eprint = {2106.05287},
  archivePrefix = {arXiv},
  primaryClass = {cond-mat.str-el}
}

@article{caceres2022controlled,
  title = {Controlled Transport Based on Multiorbital Aharonov-Bohm Photonic Caging},
  author = {C\'aceres-Aravena, Gabriel and Guzm\'an-Silva, Diego and Salinas, Ignacio and Vicencio, Rodrigo A.},
  journal = {Phys. Rev. Lett.},
  volume = {128},
  issue = {25},
  pages = {256602},
  numpages = {6},
  year = {2022},
  month = {Jun},
  publisher = {American Physical Society},
  doi = {10.1103/PhysRevLett.128.256602},
  url = {https://link.aps.org/doi/10.1103/PhysRevLett.128.256602}
}

@article{wang2022observation,
  title = {Observation of inverse Anderson transitions in Aharonov-Bohm topolectrical circuits},
  author = {Wang, Haiteng and Zhang, Weixuan and Sun, Houjun and Zhang, Xiangdong},
  journal = {Phys. Rev. B},
  volume = {106},
  issue = {10},
  pages = {104203},
  numpages = {9},
  year = {2022},
  month = {Sep},
  publisher = {American Physical Society},
  doi = {10.1103/PhysRevB.106.104203},
  url = {https://link.aps.org/doi/10.1103/PhysRevB.106.104203}
}

@article{he2021flat,
  title = {Flat-Band Localization in Creutz Superradiance Lattices},
  author = {He, Yanyan and Mao, Ruosong and Cai, Han and Zhang, Jun-Xiang and Li, Yongqiang and Yuan, Luqi and Zhu, Shi-Yao and Wang, Da-Wei},
  journal = {Phys. Rev. Lett.},
  volume = {126},
  issue = {10},
  pages = {103601},
  numpages = {7},
  year = {2021},
  month = {Mar},
  publisher = {American Physical Society},
  doi = {10.1103/PhysRevLett.126.103601},
  url = {https://link.aps.org/doi/10.1103/PhysRevLett.126.103601}
}

@article{calugaru2022general,
  author = {Călugăru, Dumitru and Chew, Aaron and Elcoro, Luis and Xu, Yuanfeng and Regnault, Nicolas and Song, Zhi-Da and Bernevig, B. Andrei},
  year = {2022},
  date = {2022/02/01},
  title ={General construction and topological classification of crystalline flat bands},
  journal = {Nature Physics},
  pages = {185--189},
  volume = {18},
  issue = {2},
  isbn = {1745-2481},
  url = {https://doi.org/10.1038/s41567-021-01445-3},
  doi = {10.1038/s41567-021-01445-3}
}

@article{vicencio2021photonic,
  author = {Rodrigo A. Vicencio Poblete},
  title = {Photonic flat band dynamics},
  journal = {Advances in Physics: X},
  volume = {6},
  number = {1},
  pages = {1878057},
  year = {2021},
  publisher = {Taylor & Francis},
  doi = {10.1080/23746149.2021.1878057},
  url = {https://doi.org/10.1080/23746149.2021.1878057},
  eprint = {https://doi.org/10.1080/23746149.2021.1878057}
}

@article{li2022aharonov,
  title = {Aharonov-Bohm Caging and Inverse Anderson Transition in Ultracold Atoms},
  author = {Li, Hang and Dong, Zhaoli and Longhi, Stefano and Liang, Qian and Xie, Dizhou and Yan, Bo},
  journal = {Phys. Rev. Lett.},
  volume = {129},
  issue = {22},
  pages = {220403},
  numpages = {6},
  year = {2022},
  month = {Nov},
  publisher = {American Physical Society},
  doi = {10.1103/PhysRevLett.129.220403},
  url = {https://link.aps.org/doi/10.1103/PhysRevLett.129.220403}
}

@article{kuno2021multiple,
  title = {Multiple quantum scar states and emergent slow thermalization in a flat-band system},
  author = {Kuno, Yoshihito and Mizoguchi, Tomonari and Hatsugai, Yasuhiro},
  journal = {Phys. Rev. B},
  volume = {104},
  issue = {8},
  pages = {085130},
  numpages = {12},
  year = {2021},
  month = {Aug},
  publisher = {American Physical Society},
  doi = {10.1103/PhysRevB.104.085130},
  url = {https://link.aps.org/doi/10.1103/PhysRevB.104.085130}
}

@article{leykam2018artificial,
 author = {Leykam, Daniel and Andreanov, Alexei and Flach, Sergej},
 doi = {10.1080/23746149.2018.1473052},
 journal = {Adv. Phys.: X},
 number = {1},
 pages = {1473052},
 publisher = {Taylor \& Francis},
 title = {Artificial flat band systems: from lattice models to experiments},
 url = {https://doi.org/10.1080/23746149.2018.1473052},
 volume = {3},
 year = {2018},
}

@article{leykam2018perspective,
 author = {Leykam, Daniel and Flach, Sergej},
 doi = {10.1063/1.5034365},
 journal = {APL Photonics},
 number = {7},
 pages = {070901},
 title = {Perspective: Photonic flatbands},
 url = {https://doi.org/10.1063/1.5034365},
 volume = {3},
 year = {2018},
}

@article{maimaiti2019universal,
 author = {Maimaiti, Wulayimu and Flach, Sergej and Andreanov, Alexei},
 title = {Universal $d=1$ flat band generator from compact localized states},
 journal = {Phys. Rev. B},
 volume = {99},
 issue = {12},
 pages = {125129},
 numpages = {15},
 year = {2019},
 month = mar,
 publisher = {American Physical Society},
 doi = {10.1103/PhysRevB.99.125129},
 url = {https://link.aps.org/doi/10.1103/PhysRevB.99.125129},
}

@article{rhim2019classification,
 author = {Rhim, Jun-Won and Yang, Bohm-Jung},
 title = {Classification of flat bands according to the band-crossing singularity of Bloch wave functions},
 journal = {Phys. Rev. B},
 volume = {99},
 issue = {4},
 pages = {045107},
 numpages = {21},
 year = {2019},
 month = jan,
 publisher = {American Physical Society},
 doi = {10.1103/PhysRevB.99.045107},
 url = {https://link.aps.org/doi/10.1103/PhysRevB.99.045107},
}

@article{dias2015origami,
 author = {Dias, R.G. and Gouveia, J.D.},
 date = {2015/11/19/online},
 day = {19},
 journal = {Sci. Rep.},
 doi = {10.1038/srep16852},
 month = nov,
 pages = {16852},
 publisher = {Springer Nature},
 title = {Origami rules for the construction of localized eigenstates of the {Hubbard} model in decorated lattices},
 url = {https://doi.org/10.1038/srep16852},
 volume = {5},
 issue = {1},
 year = {2015},
}

@article{pal2018flat,
  title = {Flat bands in fractal-like geometry},
  author = {Pal, Biplab and Saha, Kush},
  journal = {Phys. Rev. B},
  volume = {97},
  issue = {19},
  pages = {195101},
  numpages = {7},
  year = {2018},
  month = {May},
  publisher = {American Physical Society},
  doi = {10.1103/PhysRevB.97.195101},
  url = {https://link.aps.org/doi/10.1103/PhysRevB.97.195101}
}

@article{ramachandran2017chiral,
 author = {Ramachandran, Ajith and Andreanov, Alexei and Flach, Sergej},
 doi = {10.1103/PhysRevB.96.161104},
 issue = {16},
 journal = {Phys. Rev. B},
 month = oct,
 numpages = {5},
 pages = {161104(R)},
 publisher = {American Physical Society},
 title = {Chiral flat bands: Existence, engineering, and stability},
 url = {https://link.aps.org/doi/10.1103/PhysRevB.96.161104},
 volume = {96},
 year = {2017},
}

@article{vidal2000interaction,
 author = {Vidal, Julien and Dou\ifmmode \mbox{\c{c}}\else \c{c}\fi{}ot, Beno\^{\i}t and Mosseri, R\'emy and Butaud, Patrick},
 doi = {10.1103/PhysRevLett.85.3906},
 issue = {18},
 journal = {Phys. Rev. Lett.},
 month = oct,
 numpages = {0},
 pages = {3906--3909},
 publisher = {American Physical Society},
 title = {Interaction Induced Delocalization for Two Particles in a Periodic Potential},
 url = {https://link.aps.org/doi/10.1103/PhysRevLett.85.3906},
 volume = {85},
 year = {2000},
}

@article{aoki1996hofstadter,
 author = {Aoki, Hideo and Ando, Masato and Matsumura, Hajime},
 doi = {10.1103/PhysRevB.54.R17296},
 issue = {24},
 journal = {Phys. Rev. B},
 month = dec,
 numpages = {0},
 pages = {R17296--R17299},
 publisher = {American Physical Society},
 title = {Hofstadter butterflies for flat bands},
 url = {https://link.aps.org/doi/10.1103/PhysRevB.54.R17296},
 volume = {54},
 year = {1996},
}

@article{vidal2001disorder,
 author = {Vidal, Julien and Butaud, Patrick and Dou\ifmmode \mbox{\c{c}}\else \c{c}\fi{}ot, Benoit and Mosseri, R\'emy},
 doi = {10.1103/PhysRevB.64.155306},
 issue = {15},
 journal = {Phys. Rev. B},
 month = sep,
 numpages = {16},
 pages = {155306},
 publisher = {American Physical Society},
 title = {Disorder and interactions in {Aharonov-Bohm} cages},
 url = {https://link.aps.org/doi/10.1103/PhysRevB.64.155306},
 volume = {64},
 year = {2001},
}

@article{leykam2017localization,
 author = {Leykam, Daniel and Bodyfelt, Joshua D. and Desyatnikov, Anton S. and Flach, Sergej},
 doi = {10.1140/epjb/e2016-70551-2},
 issn = {1434-6036},
 journal = {Eur. Phys. J. B},
 number = {1},
 pages = {1},
 title = {Localization of weakly disordered flat band states},
 url = {https://doi.org/10.1140/epjb/e2016-70551-2},
 volume = {90},
 year = {2017},
}

@article{maimaiti2017compact,
 author = {Maimaiti, Wulayimu and Andreanov, Alexei and Park, Hee Chul and Gendelman, Oleg and Flach, Sergej},
 doi = {10.1103/PhysRevB.95.115135},
 journal = {Phys. Rev. B},
 month = mar,
 number = {11},
 pages = {115135},
 publisher = {American Physical Society},
 title = {Compact localized states and flat-band generators in one dimension},
 url = {https://link.aps.org/doi/10.1103/PhysRevB.95.115135},
 volume = {95},
 year = {2017},
}

@article{weimann2016transport,
 author = {Weimann, Steffen and Morales-Inostroza, Luis and Real, Basti\'{a}n and Cantillano, Camilo and Szameit, Alexander and Vicencio, Rodrigo A.},
 doi = {10.1364/OL.41.002414},
 journal = {Opt. Lett.},
 month = jun,
 number = {11},
 pages = {2414--2417},
 publisher = {OSA},
 title = {Transport in Sawtooth photonic lattices},
 url = {http://ol.osa.org/abstract.cfm?URI=ol-41-11-2414},
 volume = {41},
 year = {2016},
}

@article{vicencio2015observation,
 author = {Vicencio, Rodrigo A. and Cantillano, Camilo and Morales-Inostroza, Luis and Real, Basti\'an and Mej\'{\i}a-Cort\'es, Cristian and Weimann, Steffen and Szameit, Alexander and Molina, Mario I.},
 doi = {10.1103/PhysRevLett.114.245503},
 issue = {24},
 journal = {Phys. Rev. Lett.},
 month = jun,
 numpages = {5},
 pages = {245503},
 publisher = {American Physical Society},
 title = {Observation of Localized States in {Lieb} Photonic Lattices},
 url = {http://link.aps.org/doi/10.1103/PhysRevLett.114.245503},
 volume = {114},
 year = {2015},
}

@article{taie2015coherent,
 author = {Shintaro Taie  and Hideki Ozawa  and Tomohiro Ichinose  and Takuei Nishio  and Shuta Nakajima  and Yoshiro Takahashi },
 title = {Coherent driving and freezing of bosonic matter wave in an optical Lieb lattice},
 journal = {Science Advances},
 volume = {1},
 number = {10},
 pages = {e1500854},
 year = {2015},
 doi = {10.1126/sciadv.1500854},
 url = {https://www.science.org/doi/abs/10.1126/sciadv.1500854}
}

@article{mukherjee2015observation,
 author = {Mukherjee, Sebabrata and Spracklen, Alexander and Choudhury, Debaditya and Goldman, Nathan and \"Ohberg, Patrik and Andersson, Erika and Thomson, Robert R.},
 doi = {10.1103/PhysRevLett.114.245504},
 issue = {24},
 journal = {Phys. Rev. Lett.},
 month = jun,
 numpages = {5},
 pages = {245504},
 publisher = {American Physical Society},
 title = {Observation of a Localized Flat-Band State in a Photonic {Lieb} Lattice},
 url = {http://link.aps.org/doi/10.1103/PhysRevLett.114.245504},
 volume = {114},
 year = {2015},
}

@article{mielke1991ferromagnetism,
 author = {Mielke, A},
 journal = {J. Phys. A: Math. Gen.},
 number = {14},
 pages = {3311},
 title = {Ferromagnetism in the Hubbard model on line graphs and further considerations},
 url = {http://stacks.iop.org/0305-4470/24/i=14/a=018},
 volume = {24},
 year = {1991},
}

@article{khomeriki2016landau,
 author = {Khomeriki, Ramaz and Flach, Sergej},
 doi = {10.1103/PhysRevLett.116.245301},
 issue = {24},
 journal = {Phys. Rev. Lett.},
 month = jun,
 numpages = {5},
 pages = {245301},
 publisher = {American Physical Society},
 title = {Landau-Zener Bloch Oscillations with Perturbed Flat Bands},
 url = {http://link.aps.org/doi/10.1103/PhysRevLett.116.245301},
 volume = {116},
 year = {2016},
}

@article{tasaki1992ferromagnetism,
 author = {Tasaki, Hal},
 doi = {10.1103/PhysRevLett.69.1608},
 issue = {10},
 journal = {Phys. Rev. Lett.},
 month = sep,
 numpages = {0},
 pages = {1608--1611},
 publisher = {American Physical Society},
 title = {Ferromagnetism in the Hubbard models with degenerate single-electron ground states},
 url = {http://link.aps.org/doi/10.1103/PhysRevLett.69.1608},
 volume = {69},
 year = {1992},
}

@article{flach2014detangling,
 author = {Flach, Sergej and Leykam, Daniel and Bodyfelt, Joshua D. and Matthies, Peter and Desyatnikov, Anton S.},
 journal = {Europhys. Lett.},
 number = {3},
 pages = {30001},
 title = {Detangling flat bands into {Fano} lattices},
 url = {http://stacks.iop.org/0295-5075/105/i=3/a=30001},
 volume = {105},
 year = {2014},
}

@article{derzhko2015strongly,
 author = {Derzhko, Oleg and Richter, Johannes and Maksymenko, Mykola},
 doi = {10.1142/S0217979215300078},
 journal = {Int. J. Mod. Phys. B},
 number = {12},
 pages = {1530007},
 title = {Strongly correlated flat-band systems: The route from Heisenberg spins to Hubbard electrons},
 url = {http://www.worldscientific.com/doi/abs/10.1142/S0217979215300078},
 volume = {29},
 year = {2015},
}

@article{vidal1998aharonov,
 author = {Vidal, Julien and Mosseri, R\'emy and Dou\ifmmode \mbox{\c{c}}\else \c{c}\fi{}ot, Benoit},
 doi = {10.1103/PhysRevLett.81.5888},
 issue = {26},
 journal = {Phys. Rev. Lett.},
 month = dec,
 numpages = {0},
 pages = {5888--5891},
 publisher = {American Physical Society},
 title = {{Aharonov-Bohm} Cages in Two-Dimensional Structures},
 url = {http://link.aps.org/doi/10.1103/PhysRevLett.81.5888},
 volume = {81},
 year = {1998},
}

@article{sutherland1986localization,
 author = {Sutherland, Bill},
 doi = {10.1103/PhysRevB.34.5208},
 issue = {8},
 journal = {Phys. Rev. B},
 month = oct,
 numpages = {0},
 pages = {5208--5211},
 publisher = {American Physical Society},
 title = {Localization of electronic wave functions due to local topology},
 url = {https://link.aps.org/doi/10.1103/PhysRevB.34.5208},
 volume = {34},
 year = {1986},
}

@article{read2017compactly,
 author = {Read, N.},
 title = {Compactly supported Wannier functions and algebraic $K$-theory},
 journal = {Phys. Rev. B},
 volume = {95},
 issue = {11},
 pages = {115309},
 numpages = {26},
 year = {2017},
 month = mar,
 publisher = {American Physical Society},
 doi = {10.1103/PhysRevB.95.115309},
 url = {https://link.aps.org/doi/10.1103/PhysRevB.95.115309},
}

@article{gligoric2019nonlinear,
 author = {Gligori{\' c}, Goran and Beli{\v c}ev, Petra P. and Leykam, Daniel and Maluckov, Aleksandra},
 title = {Nonlinear symmetry breaking of {Aharonov-Bohm} cages},
 journal = {Phys. Rev. A},
 volume = {99},
 issue = {1},
 pages = {013826},
 numpages = {6},
 year = {2019},
 month = jan,
 publisher = {American Physical Society},
 doi = {10.1103/PhysRevA.99.013826},
 url = {https://link.aps.org/doi/10.1103/PhysRevA.99.013826},
}

@article{tovmasyan2018preformed,
 author = {Tovmasyan, Murad and Peotta, Sebastiano and Liang, Long and T\"orm\"a, P\"aivi and Huber, Sebastian D.},
 title = {Preformed pairs in flat {Bloch} bands},
 journal = {Phys. Rev. B},
 volume = {98},
 issue = {13},
 pages = {134513},
 numpages = {22},
 year = {2018},
 month = oct,
 publisher = {American Physical Society},
 doi = {10.1103/PhysRevB.98.134513},
 url = {https://link.aps.org/doi/10.1103/PhysRevB.98.134513},
}

@article{kremer2020square,
 author = {Kremer, Mark and Petrides, Ioannis and Meyer, Eric and Heinrich, Matthias and Zilberberg, Oded and Szameit, Alexander},
 title = {A square-root topological insulator with non-quantized indices realized with photonic {Aharonov-Bohm} cages},
 journal = {Nat. Comm.},
 volume = {11},
 number = {1},
 pages = {907},
 year = {2020},
 doi = {10.1038/s41467-020-14692-4},
 url = {https://doi.org/10.1038/s41467-020-14692-4},
}

@article{mukherjee2018experimental,
 author = {Mukherjee, Sebabrata and Di Liberto, Marco and \"Ohberg, Patrik and Thomson, Robert R. and Goldman, Nathan},
 title = {Experimental Observation of {Aharonov-Bohm} Cages in Photonic Lattices},
 journal = {Phys. Rev. Lett.},
 volume = {121},
 issue = {7},
 pages = {075502},
 numpages = {6},
 year = {2018},
 month = aug,
 publisher = {American Physical Society},
 doi = {10.1103/PhysRevLett.121.075502},
 url = {https://link.aps.org/doi/10.1103/PhysRevLett.121.075502},
}

@article{diliberto2019nonlinear,
 author = {Di Liberto, Marco and Mukherjee, Sebabrata and Goldman, Nathan},
 title = {Nonlinear dynamics of {Aharonov-Bohm} cages},
 journal = {Phys. Rev. A},
 volume = {100},
 issue = {4},
 pages = {043829},
 numpages = {7},
 year = {2019},
 month = oct,
 publisher = {American Physical Society},
 doi = {10.1103/PhysRevA.100.043829},
 url = {https://link.aps.org/doi/10.1103/PhysRevA.100.043829},
}

@article{tilleke2020nearest,
 author = {Tilleke, Simon and Daumann, Mirko and Dahm, Thomas},
 doi = {doi:10.1515/zna-2019-0371},
 url = {https://doi.org/10.1515/zna-2019-0371},
 title = {Nearest Neighbour Particle-Particle Interaction in Fermionic Quasi One-Dimensional Flat Band Lattices},
 journal = {Zeitschrift f\"ur Naturforschung A},
 number = {5},
 volume = {75},
 year = {2020},
 pages = {393--402},
 publisher = {De Gruyter},
 address = {Berlin, Boston},
}

@article{kuno2020flat_qs,
 author = {Kuno, Yoshihito and Mizoguchi, Tomonari and Hatsugai, Yasuhiro},
 title = {Flat band quantum scar},
 journal = {Phys. Rev. B},
 volume = {102},
 issue = {24},
 pages = {241115},
 numpages = {5},
 year = {2020},
 month = dec,
 publisher = {American Physical Society},
 doi = {10.1103/PhysRevB.102.241115},
 url = {https://link.aps.org/doi/10.1103/PhysRevB.102.241115},
}

@article{orito2021nonthermalized,
 author = {Orito, Takahiro and Kuno, Yoshihito and Ichinose, Ikuo},
 title = {Nonthermalized dynamics of flat-band many-body localization},
 journal = {Phys. Rev. B},
 volume = {103},
 issue = {6},
 pages = {L060301},
 numpages = {6},
 year = {2021},
 month = feb,
 publisher = {American Physical Society},
 doi = {10.1103/PhysRevB.103.L060301},
 url = {https://link.aps.org/doi/10.1103/PhysRevB.103.L060301},
}

@article{rhim2021singular,
 author = {Rhim, Jun-Won and Yang, Bohm-Jung},
 title = {Singular flat bands},
 journal = {Advances in Physics: X},
 volume = {6},
 number = {1},
 pages = {1901606},
 year = {2021},
 publisher = {Taylor \& Francis},
 doi = {10.1080/23746149.2021.1901606},
 url = {https://doi.org/10.1080/23746149.2021.1901606},
}

@article{danieli2021nonlinear,
 author = {Danieli, Carlo and Andreanov, Alexei and Mithun, Thudiyangal and Flach, Sergej},
 title = {Nonlinear caging in all-bands-flat lattices},
 journal = {Phys. Rev. B},
 volume = {104},
 issue = {8},
 pages = {085131},
 numpages = {13},
 year = {2021},
 month = aug,
 publisher = {American Physical Society},
 doi = {10.1103/PhysRevB.104.085131},
 url = {https://link.aps.org/doi/10.1103/PhysRevB.104.085131},
}

@article{danieli2021quantum,
 author = {Danieli, Carlo and Andreanov, Alexei and Mithun, Thudiyangal and Flach, Sergej},
 title = {Quantum caging in interacting many-body all-bands-flat lattices},
 journal = {Phys. Rev. B},
 volume = {104},
 issue = {8},
 pages = {085132},
 numpages = {10},
 year = {2021},
 month = aug,
 publisher = {American Physical Society},
 doi = {10.1103/PhysRevB.104.085132},
 url = {https://link.aps.org/doi/10.1103/PhysRevB.104.085132},
}

@article{kang2020creutz,
	doi = {10.1088/1367-2630/ab61d7},
	url = {https://doi.org/10.1088/1367-2630/ab61d7},
	year = 2020,
	month = {jan},
	publisher = {{IOP} Publishing},
	volume = {22},
	number = {1},
	pages = {013023},
	author = {Jin Hyoun Kang and Jeong Ho Han and Y Shin},
	title = {Creutz ladder in a resonantly shaken 1D optical lattice},
	journal = {New Journal of Physics}
}

@article{danieli2020many,
  title = {Many-body flatband localization},
  author = {Danieli, Carlo and Andreanov, Alexei and Flach, Sergej},
  journal = {Phys. Rev. B},
  volume = {102},
  issue = {4},
  pages = {041116},
  numpages = {6},
  year = {2020},
  month = {Jul},
  publisher = {American Physical Society},
  doi = {10.1103/PhysRevB.102.041116},
  url = {https://link.aps.org/doi/10.1103/PhysRevB.102.041116}
}

@article{graf2021designing,
  title = {Designing flat-band tight-binding models with tunable multifold band touching points},
  author = {Graf, Ansgar and Pi\'echon, Fr\'ed\'eric},
  journal = {Phys. Rev. B},
  volume = {104},
  issue = {19},
  pages = {195128},
  numpages = {21},
  year = {2021},
  month = {Nov},
  publisher = {American Physical Society},
  doi = {10.1103/PhysRevB.104.195128},
  url = {https://link.aps.org/doi/10.1103/PhysRevB.104.195128}
}

@article{hwang2021general,
  title = {General construction of flat bands with and without band crossings based on wave function singularity},
  author = {Hwang, Yoonseok and Rhim, Jun-Won and Yang, Bohm-Jung},
  journal = {Phys. Rev. B},
  volume = {104},
  issue = {8},
  pages = {085144},
  numpages = {16},
  year = {2021},
  month = {Aug},
  publisher = {American Physical Society},
  doi = {10.1103/PhysRevB.104.085144},
  url = {https://link.aps.org/doi/10.1103/PhysRevB.104.085144}
}

@article{chase2024compact,
  title = {Compact localized states in electric circuit flat-band lattices},
  author = {Chase-Mayoral, Carys and English, L. Q. and Lape, Noah and Kim, Yeongjun and Lee, Sanghoon and Andreanov, Alexei and Flach, Sergej and Kevrekidis, P. G.},
  journal = {Phys. Rev. B},
  volume = {109},
  issue = {7},
  pages = {075430},
  numpages = {9},
  year = {2024},
  month = {Feb},
  publisher = {American Physical Society},
  doi = {10.1103/PhysRevB.109.075430},
  url = {https://link.aps.org/doi/10.1103/PhysRevB.109.075430}
}

@book{shankar_1994_PrinciplesQuantum,
  title = {Principles of {{Quantum Mechanics}}},
  author = {Shankar, R.},
  year = {1994},
  publisher = {Springer US},
  location = {New York, NY},
  doi = {10.1007/978-1-4757-0576-8},
  isbn = {978-1-4757-0578-2 978-1-4757-0576-8}
}

@article{han_2024_Largequantum,
  title = {Large Quantum Anomalous {{Hall}} Effect in Spin-Orbit Proximitized Rhombohedral Graphene},
  author = {Han, Tonghang and Lu, Zhengguang and Yao, Yuxuan and Yang, Jixiang and Seo, Junseok and Yoon, Chiho and Watanabe, Kenji and Taniguchi, Takashi and Fu, Liang and Zhang, Fan and Ju, Long},
  date = {2024-05-10},
  year = {2024},
  journal = {Science},
  publisher = {American Association for the Advancement of Science},
  doi = {10.1126/science.adk9749}
}

@article{lu_2025_Extendedquantum,
  title = {Extended Quantum Anomalous {{Hall}} States in Graphene/{{hBN}} Moiré Superlattices},
  author = {Lu, Zhengguang and Han, Tonghang and Yao, Yuxuan and Hadjri, Zach and Yang, Jixiang and Seo, Junseok and Shi, Lihan and Ye, Shenyong and Watanabe, Kenji and Taniguchi, Takashi and Ju, Long},
  date = {2025-01},
  year = {2025},
  journal = {Nature},
  volume = {637},
  number = {8048},
  pages = {1090--1095},
  publisher = {Nature Publishing Group},
  issn = {1476-4687},
  doi = {10.1038/s41586-024-08470-1}
}

@misc{chase-mayoral_2023_compact,
  title = {Compact {{Localized States}} in {{Electric Circuit Flatband Lattices}}},
  author = {{Chase-Mayoral}, Carys and English, L. Q. and Kim, Yeongjun and Lee, Sanghoon and Lape, Noah and Andreanov, Alexei and Kevrekidis, P. G. and Flach, Sergej},
  year = {2023},
  month = aug,
  number = {arXiv:2307.15319},
  eprint = {2307.15319},
  primaryclass = {cond-mat, physics:nlin},
  publisher = {arXiv},
  doi = {10.48550/arXiv.2307.15319},
  urldate = {2024-02-21},
  archiveprefix = {arXiv}
}

@article{hart2020compact,
 author = {Hart, Oliver and De Tomasi, Giuseppe and Castelnovo, Claudio},
 title = {From compact localized states to many-body scars in the random quantum comb},
 journal = {Phys. Rev. Research},
 volume = {2},
 issue = {4},
 pages = {043267},
 numpages = {10},
 year = {2020},
 month = nov,
 publisher = {American Physical Society},
 doi = {10.1103/PhysRevResearch.2.043267},
 url = {https://link.aps.org/doi/10.1103/PhysRevResearch.2.043267},
}

@article{pelegri2024few,
  title = {Few-body bound topological and flat-band states in a Creutz ladder},
  author = {Pelegr\'{\i}, G. and Flannigan, S. and Daley, A. J.},
  journal = {Phys. Rev. B},
  volume = {109},
  issue = {23},
  pages = {235412},
  numpages = {11},
  year = {2024},
  month = {Jun},
  publisher = {American Physical Society},
  doi = {10.1103/PhysRevB.109.235412},
  url = {https://link.aps.org/doi/10.1103/PhysRevB.109.235412}
}

@article{vicencio2025multi,
    author = {Vicencio, Rodrigo A.},
    title = {Multi-orbital photonic lattices},
    journal = {APL Photonics},
    volume = {10},
    number = {7},
    pages = {071101},
    year = {2025},
    month = {07},
    issn = {2378-0967},
    doi = {10.1063/5.0273569},
    url = {https://doi.org/10.1063/5.0273569}
}

@article{kuno2020flat22,
 author = {Kuno, Yoshihito and Orito, Takahiro and Ichinose, Ikuo},
 doi = {10.1088/1367-2630/ab6352},
 url = {https://doi.org/10.1088/1367-2630/ab6352},
 year = {2020},
 month = jan,
 publisher = {{IOP} Publishing},
 volume = {22},
 number = {1},
 pages = {013032},
 title = {Flat-band many-body localization and ergodicity breaking in the Creutz ladder},
 journal = {New J. Phys.}
}

@book{wigner_2012_group,
  title={Group theory: and its application to the quantum mechanics of atomic spectra},
  author={Wigner, Eugene},
  volume={5},
  publisher={Elsevier},
year = {1959},
issn = {0079-8193},
doi = {https://doi.org/10.1016/B978-0-12-750550-3.50001-X}
}

@article{heinzner_2005_symmetry,
  title = {Symmetry {{Classes}} of {{Disordered Fermions}}},
  author = {Heinzner, P. and Huckleberry, A. and Zirnbauer, M.R.},
  year = {2005},
  month = aug,
  journal = {Commun. Math. Phys.},
  volume = {257},
  number = {3},
  pages = {725--771},
  issn = {1432-0916},
  doi = {10.1007/s00220-005-1330-9},
  urldate = {2025-07-25}
}

@article{cordova_2018_timereversal,
  title = {Time-Reversal Symmetry, Anomalies, and Dualities in (2+1)\$d\$},
  author = {C{\'o}rdova, Clay and Hsin, Po-Shen and Seiberg, Nathan and Hsin, Po-Shen},
  year = {2018},
  month = jul,
  journal = {SciPost Physics},
  volume = {5},
  number = {1},
  pages = {006},
  issn = {2542-4653},
  doi = {10.21468/SciPostPhys.5.1.006},
  urldate = {2025-07-25}
}

@article{leykam_2024_flat,
  title = {Flat Bands, Sharp Physics},
  author = {Leykam, Daniel},
  year = {2024},
  month = jan,
  journal = {AAPPS Bulletin},
  volume = {34},
  number = {1},
  pages = {2},
  issn = {2309-4710},
  doi = {10.1007/s43673-023-00113-3},
  urldate = {2024-03-29}
}

@article{maimaiti2021flat,
  title = {Flat-band generator in two dimensions},
  author = {Maimaiti, Wulayimu and Andreanov, Alexei and Flach, Sergej},
  journal = {Phys. Rev. B},
  volume = {103},
  issue = {16},
  pages = {165116},
  numpages = {15},
  year = {2021},
  month = {Apr},
  publisher = {American Physical Society},
  doi = {10.1103/PhysRevB.103.165116},
  url = {https://link.aps.org/doi/10.1103/PhysRevB.103.165116}
}

@article{morales-inostroza_2016_simple,
  title = {Simple Method to Construct Flat-Band Lattices},
  author = {{Morales-Inostroza}, Luis and Vicencio, Rodrigo A.},
  year = {2016},
  month = oct,
  journal = {Physical Review A},
  volume = {94},
  number = {4},
  pages = {043831},
  publisher = {American Physical Society},
  doi = {10.1103/PhysRevA.94.043831},
  urldate = {2024-02-28}
}

@article{rontgen_2018_compact,
  title = {Compact Localized States and Flat Bands from Local Symmetry Partitioning},
  author = {R{\"o}ntgen, M. and Morfonios, C. V. and Schmelcher, P.},
  year = {2018},
  month = jan,
  journal = {Physical Review B},
  volume = {97},
  number = {3},
  pages = {035161},
  publisher = {American Physical Society},
  doi = {10.1103/PhysRevB.97.035161},
  urldate = {2024-02-28}
}

@article{Jorg2020artificial,
	title = {Artificial gauge field switching using orbital angular momentum modes in optical waveguides},
	author = {Jorg, Christina and Queraltó, Gerard and Kremer, Mark and Pelegrí, Gerard and Schulz, Julian and Szameit, Alexander and von Freymann, Georg and Mompart, Jordi and Ahufinger, Verònica},
	journal = {Light: Science \& Applications},
	volume = {9},
	issue = {1},
	pages = {2047-7538},
	year = {2020},
	month = {Aug},
	doi = {10.1038/s41377-020-00385-6},
	url = {https://doi.org/10.1038/s41377-020-00385-6}
}

@article{roman2025observation,
  title = {Observation of Invisibility Angle and Flat Band Physics in Dipolar Photonic Lattices},
  author = {Roman-Cortes, Diego and Mazanov, Maxim and Vicencio, Rodrigo A. and Gorlach, Maxim A.},
  journal = {Nano Letters},
  volume = {25},
  issue = {11},
  pages = {4291},
  year = {2025},
  publisher = {American Chemical Society},
  doi = {doi: 10.1021/acs.nanolett.4c05951},
  url = {https://doi.org/10.1021/acs.nanolett.4c05951}
}

@article{brosco2021two,
  title = {Two-flux tunable Aharonov-Bohm effect in a photonic lattice},
  author = {Brosco, V. and Pilozzi, L. and Conti, C.},
  journal = {Phys. Rev. B},
  volume = {104},
  issue = {2},
  pages = {024306},
  numpages = {6},
  year = {2021},
  month = {Jul},
  publisher = {American Physical Society},
  doi = {10.1103/PhysRevB.104.024306},
  url = {https://link.aps.org/doi/10.1103/PhysRevB.104.024306}
}

@incollection{gell-mann_2000_eightfold,
  title = {The {{Eightfold Way}}: {{A Theory}} of {{Strong Interaction Symmetry}}*},
  shorttitle = {The {{Eightfold Way}}},
  booktitle = {The {{Eightfold Way}}},
  author = {{Gell-Mann}, Murray},
  year = {2000},
  publisher = {CRC Press},
  isbn = {978-0-429-49661-5}
}

@article{zhou_2025_anyonic,
  title = {Anyonic Topological Flat Bands},
  author = {Zhou, Xiaoqi and Zhang, Weixuan and Zhang, Xiangdong},
  date = {2025-04-15},
  year = {2025},
  journal = {Frontiers of Physics},
  shortjournal = {Front. Phys.},
  volume = {20},
  number = {2},
  pages = {024205},
  issn = {2095-0462},
  doi = {10.15302/frontphys.2025.024205}
}

@article{zhang_2023_nonabelian,
  title = {Non-{{Abelian Inverse Anderson Transitions}}},
  author = {Zhang, Weixuan and Wang, Haiteng and Sun, Houjun and Zhang, Xiangdong},
  date = {2023-05-16},
  year = {2023},
  journal = {Physical Review Letters},
  shortjournal = {Phys. Rev. Lett.},
  volume = {130},
  number = {20},
  pages = {206401},
  publisher = {American Physical Society},
  doi = {10.1103/PhysRevLett.130.206401}
}

@article{chen_2023_decoding,
  title = {Decoding Flat Bands from Compact Localized States},
  author = {Chen, Yuge and Huang, Juntao and Jiang, Kun and Hu, Jiangping},
  date = {2023-12-30},
  year = {2023},
  journal = {Science Bulletin},
  shortjournal = {Science Bulletin},
  volume = {68},
  number = {24},
  pages = {3165--3171},
  issn = {2095-9273},
  doi = {10.1016/j.scib.2023.11.032}
}

@misc{danieli_2026_progress,
  title = {Progress on Artificial Flat Band Systems: Classifying, Perturbing, Applying},
  author = {Danieli, Carlo and Flach, Sergej},
  year = {2026},
  month = mar,
  eprint = {2603.04248},
  publisher = {arXiv},
  archiveprefix = {arXiv},
  eprintclass = {cond-mat.mes-hall},
  doi = {10.48550/arXiv.2603.04248}
}

@article{biswas_2023_designer,
  title = {Designer Quantum States on a Fractal Substrate: {{Compact}} Localization, Flat Bands and the Edge Modes},
  shorttitle = {Designer Quantum States on a Fractal Substrate},
  author = {Biswas, Sougata and Chakrabarti, Arunava},
  date = {2023-09-01},
  year = {2023},
  journal = {Physica E: Low-dimensional Systems and Nanostructures},
  shortjournal = {Physica E: Low-dimensional Systems and Nanostructures},
  volume = {153},
  pages = {115762},
  issn = {1386-9477},
  doi = {10.1016/j.physe.2023.115762}
}

@article{santos_2020_methods,
  title = {Methods for the Construction of Interacting Many-Body {{Hamiltonians}} with Compact Localized States in Geometrically Frustrated Clusters},
  author = {Santos, F. D. R. and Dias, R. G.},
  date = {2020-03-11},
  year = {2020},
  journal = {Scientific Reports},
  shortjournal = {Sci Rep},
  volume = {10},
  number = {1},
  pages = {4532},
  publisher = {Nature Publishing Group},
  issn = {2045-2322},
  doi = {10.1038/s41598-020-60975-7}
}

@article{yang_2020_PhotonicFloquet,
  title = {Photonic {{Floquet}} Topological Insulators in a Fractal Lattice},
  author = {Yang, Zhaoju and Lustig, Eran and Lumer, Yaakov and Segev, Mordechai},
  date = {2020-07-20},
  year = {2020},
  journal = {Light: Science \& Applications},
  shortjournal = {Light Sci Appl},
  volume = {9},
  number = {1},
  pages = {128},
  publisher = {Nature Publishing Group},
  issn = {2047-7538},
  doi = {10.1038/s41377-020-00354-z}
}

@article{rechtsman_2013_PhotonicFloquet,
  title = {Photonic {{Floquet}} Topological Insulators},
  author = {Rechtsman, Mikael C. and Zeuner, Julia M. and Plotnik, Yonatan and Lumer, Yaakov and Podolsky, Daniel and Dreisow, Felix and Nolte, Stefan and Segev, Mordechai and Szameit, Alexander},
  date = {2013-04},
  year = {2013},
  journal = {Nature},
  volume = {496},
  number = {7444},
  pages = {196--200},
  publisher = {Nature Publishing Group},
  issn = {1476-4687},
  doi = {10.1038/nature12066}
}

@article{lustig_2019_Photonictopological,
  title = {Photonic Topological Insulator in Synthetic Dimensions},
  author = {Lustig, Eran and Weimann, Steffen and Plotnik, Yonatan and Lumer, Yaakov and Bandres, Miguel A. and Szameit, Alexander and Segev, Mordechai},
  date = {2019-03},
  year = {2019},
  journal = {Nature},
  volume = {567},
  number = {7748},
  pages = {356--360},
  publisher = {Nature Publishing Group},
  issn = {1476-4687},
  doi = {10.1038/s41586-019-0943-7}
}

@article{wang_2020_Circuitimplementation,
  title = {Circuit Implementation of a Four-Dimensional Topological Insulator},
  author = {Wang, You and Price, Hannah M. and Zhang, Baile and Chong, Y. D.},
  date = {2020-05-12},
  year = {2020},
  journal = {Nature Communications},
  shortjournal = {Nat Commun},
  volume = {11},
  number = {1},
  pages = {2356},
  publisher = {Nature Publishing Group},
  issn = {2041-1723},
  doi = {10.1038/s41467-020-15940-3}
}

\end{document}